\documentclass[11pt]{amsart}
\title{Spectral inequalities for Jacobi operators and related sharp Lieb--Thirring inequalities on the continuum}
\author{Lukas Schimmer}
\address{Imperial College London, 180 Queen's Gate, London SW7 2AZ, United Kingdom}
\email{l.schimmer11@imperial.ac.uk}

\usepackage{fullpage}
\usepackage[latin1]{inputenc}
\usepackage[english]{babel}
\usepackage{latexsym}
\usepackage{amsmath}
\usepackage{amssymb}
\usepackage{bbm}
\usepackage{enumitem}
\usepackage{array}
\usepackage{ifthen}
\usepackage{graphicx}
\usepackage{hyperref}   

\newcommand{\norm}[2][]{\left\|#2\right\|_{#1}}   
\newcommand{\sklp}[2][]{\left\langle#2\right\rangle_{#1}} 

\newcommand{\N}{\mathbb{N}}
\newcommand{\Z}{\mathbb{Z}}

\newcommand{\R}{\mathbb{R}}
\newcommand{\C}{\mathbb{C}}

\newcommand{\id}{\mathbb{I}}
\newcommand{\Oo}{\mathcal{O}}

\newcommand{\e}{\mathrm{e}}

\newcommand{\dom}{\mathrm{dom}} 
\newcommand{\tr}{\mathrm{tr}} 
\newcommand{\dd}{\, \mathrm{d}}
\newcommand{\cl}{\mathrm{cl}} 
\newcommand{\B}{\mathrm{B}}

\newcounter{satz}[section]
\renewcommand{\thesatz}{\arabic{section}.\arabic{satz}}

\newenvironment{theorem}[1][]
{\bigskip\noindent\refstepcounter{satz}\textsc{Theorem~\thesatz}\ifthenelse{\equal{#1}{}}{.}{~(\textit{#1}):} \, \it}
{\bigskip}

\newenvironment{lemma}[1][]
{\bigskip\noindent\refstepcounter{satz}\textsc{Lemma~\thesatz}\ifthenelse{\equal{#1}{}}{.}{~(\textit{#1}):} \, \it}
{\bigskip}

\newenvironment{proposition}[1][]
{\bigskip\noindent\refstepcounter{satz}\textsc{Proposition~\thesatz}\ifthenelse{\equal{#1}{}}{.}{~(\textit{#1}):} \, \it}
{\bigskip}

\begin{document}

\begin{abstract}
In this paper we approximate a Schr{\"o}dinger operator on $L^2(\R)$ by Jacobi operators on $\ell^2(\Z)$ to provide new proofs of sharp Lieb--Thirring inequalities for the powers $\gamma=\frac12$ and $\gamma=\frac32$. To this end we first investigate spectral inequalities for Jacobi operators. Using the commutation method we present a new, direct proof of a sharp inequality corresponding to a Lieb--Thirring inequality for the power $\frac32$ on $\ell^2(\Z)$.  
We also introduce inequalities for higher powers of the eigenvalues as well as for matrix-valued potentials and compare our results to previously established bounds. 
\end{abstract}

\maketitle

Consider the Schr{\"o}dinger operator $H=-\Delta+V$ on  $L^2(\R^d)$ with real-valued potential $V:\R^d\to\R$ satisfying $\lim_{|x|\to\infty}V(x)=0$. Denote by $\mu_1\le\mu_2\le\dots\le0$ the (not necessarily finite) sequence of  negative eigenvalues of $H$. In \cite{LiebThirring1976} Lieb and Thirring showed that, if $\gamma>\max(0,1-\frac d2)$, there exists a constant $L_{\gamma,d}$ independent of $V$ such that
\begin{align}
\tr\left(-\Delta+V\right)^{\gamma}_-=\sum_{j}|\mu_j|^\gamma\le L_{\gamma,d}\int_{\R^d}V(x)_-^{\gamma+\frac d2}\dd x\,.
\label{eq:lieb} 
\end{align}
Here and below, $a_\pm=(|a|\pm a)/2$ denotes the positive/negative part of a variable, a function or a self-adjoint operator. 
Subsequently \eqref{eq:lieb} was extended to the case $\gamma=0$, $d\ge3$ (\cite{Cwikel1977}, \cite{Lieb1976}, \cite{Rozenblum1976}) and the remaining critical case $\gamma=\frac12$, $d=1$ was proven to hold by Weidl in \cite{Weidl1996}.

It is of interest to identify the value of the sharp constant $L_{\gamma,d}$. The semi-classical Lieb--Thirring constant $L_{\gamma,d}^{\cl}$  is defined as
\begin{align*}
L_{\gamma,d}^{\cl}=(4\pi)^{-\frac d2}\frac{\Gamma(\gamma+1)}{\Gamma(\gamma+1+\frac d2)}
\end{align*}
and a Weyl-type asymptotic result for the left-hand-side of \eqref{eq:lieb} implies that $L_{\gamma,d}^{\cl}\le L_{\gamma,d}$. In \cite{LiebThirring1976}, Lieb and Thirring were able to show that  $ L_{\gamma,d}=L_{\gamma,d}^{\cl}$ in the case $\gamma\ge\frac32$, $d=1$. Their proof is based on the Buslaev--Faddeev--Zakharov trace formulae  for Schr{\"o}dinger operators \cite{Buslaev1960,Zakharov1971}. The Aizenman--Lieb principle \cite{Aizenman1978} proved that for fixed dimension $d$ the ratio $L_{\gamma,d}/L_{\gamma,d}^{\cl}$ is a monotone, non-increasing function of $\gamma$.  Subsequently $L_{\gamma,d}=L_{\gamma,d}^{\cl}$ was also shown to hold in the case $\gamma\ge\frac32$, $d\ge1$ by Laptev and Weidl \cite{Laptev2000}. The authors first established Buslaev--Faddeev--Zakharov trace formulae in the more general case of $H$ being defined on $L^2(\R,\C^m)$, where the potential is a Hermitian matrix-valued function $V:\R\to\C^{m\times m}$. This yielded sharp Lieb--Thirring inequalities in one dimension for matrix-valued potentials and  $\gamma\ge\frac32$. An induction argument making use of the matrix nature of the results lifted the inequalities to higher dimensions.    

For $\gamma=\frac12,\,d=1$ Hundertmark, Lieb and Thomas \cite{Hundertmark1998} proved that the best constant is given by $L_{1/2,1}=2L_{1/2,1}^{\cl}=\frac12$. The inequality is then sharp for delta potentials. 
In \cite{Eden1991} it was shown that for $1\le \gamma<\frac32$ the sharp constant can be bounded by $L_{\gamma,1}\le \frac{\pi}{\sqrt{3}}L_{\gamma,1}^{\cl}$, which was generalised to higher dimensions in \cite{Dolbeault2008}.

In our paper we provide a new proof for sharp Lieb--Thirring bounds in one dimension for $\gamma=\frac12$ and $\gamma=\frac32$ using an approximation of the Schr{\"o}dinger operator by Jacobi operators on the lattice. 
The self-adjoint Jacobi operator $W$ on $\ell^2(\Z)$  is defined as
\begin{align}
(Wu)(n)=a(n-1)u(n-1)+a(n)u(n+1)+b(n)u(n),\,n\in\Z\,.
\label{eq:defH} 
\end{align}
The sequence $a:\Z\to\R^-$ is assumed to be negative with $\lim_{|n|\to\infty}a(n)=-1$. Motivated by the representation of $W$ as a doubly infinite matrix, the numbers $a(n)$ are referred to as the off-diagonal entries of $W$. The potential $b:\Z\to\R$ is supposed to be real-valued with $\lim_{|n|\to\infty}b(n)=0$. 
If $a(n)=-1$ for all $n\in\Z$ we call $W$ a discrete Schr{\"o}dinger operator. Note that when aiming to obtain sharp Lieb--Thirring inequalities on the continuum, it is not sufficient to only consider discrete Schr{\"o}dinger operators on $\ell^2(\Z)$ to approximate Schr{\"o}dinger operators on $L^2(\R)$. As we will see later, it is essential to derive spectral inequalities for general off-diagonal entries $a(n)$ which are not constant $-1$.  The only exception is the case $\gamma=\frac12$.

Let the numbers 
$|\lambda_1|\ge|\lambda_2|\ge\dots\ge2$ denote the (not necessarily finite) sequence of eigenvalues of $W$ outside the essential spectrum $[-2,2]$, arranged by decreasing modulus.
In 2002 Hundertmark and Simon \cite{Hundertmark2002} proved that 
\begin{align}
\sum_{j}\big(|\lambda_j|^2-4\big)^{\frac12}
\le \sum_{n\in\Z}|b(n)|+4\sum_{n\in\Z}|a(n)+1|
\label{eq:hsmain} 
\end{align}
provided that the right-hand-side is finite. The authors first established the inequality in the special case of discrete Schr{\"o}dinger operators where, in complete analogy to the case of a Schr{\"o}dinger operator on the continuum, it was proven to be sharp for delta potentials. Hundertmark and Simon then used a perturbation argument to extend the result to more general off-diagonal entries $a(n)$. The constants in front of the two sums on the right-hand-side (1 and 4) were both found to be optimal. We will show that \eqref{eq:hsmain} with constant $a\equiv-1$ is sufficient to prove the sharp Lieb--Thirring inequality on $L^2(\R)$ for $\gamma=\frac12$.  To give further merit to the fact that the result of Hundertmark and Simon is also asymptotically optimal for general off-diagonal entries, we will show that restricting \eqref{eq:hsmain} to vanishing potentials $b\equiv0$ also yields the sharp Lieb--Thirring constant $L_{1/2,1}$ on $L^2(\R)$. 

To obtain analogous inequalities to \eqref{eq:hsmain} for higher powers of the eigenvalues, Hundertmark and Simon applied the Aizenman-Lieb principle \cite{Aizenman1978}
which requires the inequality to be written in terms of $\mathrm{dist}(\lambda_j,[-2,2])$. This can be done using 
$|\lambda_j|^2-4\ge4(|\lambda_j|-2)$ and for powers $\gamma\ge\frac12$ 
it yields the Lieb--Thirring inequalities
\begin{align}
\sum_{j}\big(|\lambda_j|-2\big)^{\gamma}
\le c_\gamma\left(\sum_{n\in\Z}|b(n)|^{\gamma+\frac12}+4\sum_{n\in\Z}|a(n)+1|^{\gamma+\frac12}\right)
\label{eq:hs1} 
\end{align}
with the constants 
\begin{align*}
c_\gamma=3^{\gamma-\frac12}\frac12 \frac{\Gamma(\gamma+1)\Gamma(2)}{\Gamma(\gamma+\frac32)\Gamma(\frac32)}\,.
\end{align*}
In the special case where $a(n)=-1$ for all $n\in\Z$, the constant $c_\gamma$ can be replaced by the smaller value $d_\gamma=3^{\frac12-\gamma}c_\gamma$. This stems from the fact that Hundertmark and Simon again first established the inequalities for discrete Schr{\"o}dinger operators and then used a perturbation argument to extend them to more general off-diagonal entries $a(n)$. As shown later in \cite{Sahovic2010}, the constants $c_\gamma$ and $d_\gamma$ can be improved by a factor $\frac{\pi}{2\sqrt{3}}$ if $\gamma\ge1$. 
If we let $\eta>0$ and consider $W$ with a potential $\eta b$ then the ratio of the left-hand-side and the right-hand-side of \eqref{eq:hsmain} converges to 1 for $\eta\to\infty$. In \eqref{eq:hs1}, however, the left-hand-side grows like $\eta^\gamma$, whereas the right-hand-side behaves like $\eta^{\gamma+\frac12}$, suggesting that the bounds are not optimal for large eigenvalues.  

In \cite{Hundertmark2002} Hundertmark and Simon used the inequality $(\lambda_j)^2-4\ge(|\lambda_j|-2)^2$ and the Aizenman--Lieb principle to show that \eqref{eq:hsmain} also implies 
\begin{align}
\sum_{j}\big(|\lambda_j|-2\big)^{\gamma+\frac12}
\le 3^{\gamma-\frac12}\left(\sum_{n\in\Z}|b(n)|^{\gamma+\frac12}+4\sum_{n\in\Z}|a(n)+1|^{\gamma+\frac12}\right)
\label{eq:hs2} 
\end{align}
where $\gamma\ge\frac12$. If $a(n)=-1$ for all $n\in\Z$, the factor $3^{\gamma-\frac12}$ can be omitted.
Inequality \eqref{eq:hs2} has the desired property that for potentials $\eta b$ the ratio of the left-hand-side and the right-hand-side converges to 1 as $\eta\to\infty$. It is, however, not optimal for small coupling. In fact, for eigenvalues close to the essential spectrum, the left-hand-side of inequality \eqref{eq:hs1} is much larger than the left-hand-side of \eqref{eq:hs2}, while the corresponding right-hand-sides only differ by a multiplicative constant. Only \eqref{eq:hsmain} is optimal for both large and small coupling.  
Note that  in \cite{Hundertmark2002} the inequalities \eqref{eq:hs1} and \eqref{eq:hs2} were also generalised to matrix-valued potentials and hence, by applying the method developed by Laptev and Weidl \cite{Laptev2000}, also to $d$-dimensional Jacobi operators.   

To obtain optimal estimates for both small and large coupling, different types of inequalities have to be introduced.
In \cite{Killip2003} Killip and Simon showed that if we write the eigenvalues as $\lambda_j=-k_j-\frac{1}{k_j}$ with $|k_j|>1$, then 
\begin{align}
\sum_{j}\left( k_j^2-\frac{1}{k_j^2}-\log |k_j|^4\right)\le\sum_{n\in\Z}b(n)^2+2\sum_{n\in\Z} \left(a(n)^2-1-\log a(n)^2\right)\,.
\label{eq:killipsimon} 
\end{align}
Note the emergence of logarithmic terms in both sides of this inequality. The term $\log |k_j|^4$ is essential to guarantee good estimates for both small and large coupling, as we will discuss later. We will prove that the inequality is sharp for reflectionless potentials. For all known reflectionless potentials it holds that $a\not\equiv-1$, leading to the conclusion that in contrast to \eqref{eq:hsmain} the inequality is not necessarily sharp when restricted to discrete Schr{\"o}dinger operators. 
So far no connection has been made in the literature between \eqref{eq:hsmain} and \eqref{eq:killipsimon}. In our paper we will show that at least in the case of a discrete Schr{\"o}dinger operator, the latter can be obtained from the former using an adaptation of the Aizenman--Lieb principle. This is especially interesting as both inequalities are optimal for large coupling. In contrast, a similar argument in the case of a Schr{\"o}dinger operator on the continuum does not allow to deduce the optimal Lieb--Thirring inequality for $\gamma=\frac32$ from the case of $\gamma=\frac12$. 
Using the adapted Aizenman--Lieb principle, we will also prove inequalities for discrete Schr{\"o}dinger operators for any $\gamma\ge\frac12$. The results are better than both \eqref{eq:hs1} and \eqref{eq:hs2}. 

In our paper we will show that \eqref{eq:killipsimon} yields the sharp Lieb--Thirring constant $L_{3/2,1}$ on $L^2(\R)$. Unlike in the proof for $\gamma=\frac12$, it will not be sufficient to restrict ourselves to discrete Schr{\"o}dinger operators in the approximation of the Schr{\"o}dinger operator on the continuum. The terms depending on $b$ and the terms depending on $a$ in \eqref{eq:killipsimon} will play an equally important role. For this reason, the  inequalities for $\frac12<\gamma<\frac32$, which we only managed to prove for discrete Schr{\"o}dinger operators, will not yield the sharp Lieb--Thirring constants $L_{\gamma,1}$ on $L^2(\R)$.   
 
In their paper \cite{Killip2003}, Killip and Simon derived \eqref{eq:killipsimon} from an identity
obtained from Case's sum rules (first stated in \cite{Case1975}), for which they provided two different proofs.  The first proof made use of the perturbation determinant and the Jost function of $W$, while the second proof involved a continuous fraction expansion of the $m$-function.
Killip and Simon used the sum rules to not only show \eqref{eq:killipsimon} but also to prove a characterisation of Hilbert--Schmidt perturbations of the free Jacobi matrix and results related to the Szeg\H{o} condition. Using the same methods, \eqref{eq:killipsimon} was generalised to matrix-valued potentials $B$ in \cite{Simon2011}. 
Note that, for the case of a discrete Schr{\"o}dinger operator, inequality \eqref{eq:killipsimon}  can also be found in earlier work by Deift and Killip \cite{Deift1999}. In their paper, the inequality is obtained from a trace identity for $W$.

In our work, we aim to prove \eqref{eq:killipsimon} directly, without the use of Case's sum rules. To do so, we will apply the commutation method, which goes back to the idea of inserting eigenvalues into the spectrum of differential operators as discussed by Jacobi \cite{Jacobi1837}, Darboux \cite{Darboux1882} and Crum \cite{Crum1955}. For more details on this method and some of its applications we refer to \cite{Deift1978} and \cite{Gesztesy1993} as well as \cite{Gesztesy1996}, which specifically considers the case of Jacobi operators. 
The commutation method was used in \cite{Schmincke1978} where the Schmincke inequality was established for a Schr{\"o}dinger operator on $L^2(\R)$. The same method can also be applied to Schr{\"o}dinger operators with matrix-valued  potentials. In \cite{Benguria2000} Benguria and Loss presented an elementary proof of sharp Lieb--Thirring inequalities for Schr{\"o}dinger operators on the continuum with matrix-valued potentials based on the commutation method. This provided a new proof of the inequality that had previously been established using trace identities by Laptev and Weidl in \cite{Laptev2000}. Similar results were obtained on the half-line in \cite{Exner2013}. In our work we will also consider matrix-valued potentials. 

This paper is organised as follows. In Section \ref{sec:main} we will state our main results for scalar Jacobi operators as well as for the case of matrix-valued potentials. In the subsequent sections the proofs of these statements are given. In Section \ref{sec:sharp} it is shown that the obtained inequalities are sharp for an explicit choice of $W$ with reflectionless potential. In Section \ref{sec:higher} we will show that at least  in the case  $a\equiv-1$ inequality \eqref{eq:killipsimon} can be obtained from the main result of Hundertmark and Simon \eqref{eq:hsmain}. We will use the same method to derive similar inequalities for higher powers $\gamma>\frac12$ of the eigenvalues and the potential, restricting ourselves again to the case of a discrete Schr{\"o}dinger operator. These inequalities will also include logarithmic terms and we will show that they are more precise than \eqref{eq:hs1} and \eqref{eq:hs2}.
In Section \ref{sec:approx1} we will discuss the approximation of a Schr{\"o}dinger operator on the continuum by discrete Schr{\"o}dinger operators and we will use inequality \eqref{eq:hsmain} to provide a new proof for the sharp Lieb--Thirring constant $L_{1/2,1}$ on $L^2(\R)$. 
In Section \ref{sec:approx2} we will use a similar technique involving more general Jacobi operators and inequality \eqref{eq:killipsimon} to prove sharp Lieb--Thirring bounds for $\gamma\ge\frac32$. 

\section{Statement of the main results for Jacobi operators} \label{sec:main}

We assume that $a(n)$ and $b(n)$ are both uniformly bounded for $n\in\Z$ such that \eqref{eq:defH} defines a self-adjoint operator $W$ with domain $\ell^2(\Z)$. 

\begin{theorem}\label{th:main}
Assume $b\in\ell^2(\Z)$ and that $\sum_{n\in\Z} \big(a(n)^2-1-\log a(n)^2\big)$ is finite. Then it holds that
\begin{align}
\sum_{j}\left( k_j^2-\frac{1}{k_j^2}-\log |k_j|^4\right)\le\sum_{n\in\Z}b(n)^2+2\sum_{n\in\Z} \left(a(n)^2-1-\log a(n)^2\right)
\label{eq:final} 
\end{align}
and this inequality is sharp.
In particular, the discrete Schr{\"o}dinger operator with $a(n)=-1$ for all $n\in\Z$ and square-summable potential $b$ satisfies
\begin{align*}
\sum_{j}\left( k_j^2-\frac{1}{k_j^2}-\log |k_j|^4\right)\le\sum_{n\in\Z}b(n)^2\,.
\end{align*}

\end{theorem}
A proof of this theorem based on the commutation method will be given in  Section \ref{sec:proof}. 
A similar result can be shown to hold for Jacobi operators with matrix-valued potentials. On $\ell^2(\Z,\C^m)$ we consider the operator 
\begin{align}
(Wu)(n)=A(n-1)^*u(n-1)+A(n)u(n+1)+B(n)u(n)
\label{eq:defHmatrix} 
\end{align}
where the off-diagonal entries $A(n)\in\C^{m\times m}$ are invertible matrices for all $n\in\Z$ and the potential $B(n)\in\C^{m\times m}$ is Hermitian. We furthermore assume that $A(n)\to-\id$ as $|n|\to\infty$ and that $\norm{A(n)}$ and $\norm{B(n)}$ are uniformly bounded for $n\in\Z$. The latter guarantees that the operator $W$ is well-defined on the domain $\ell^2(\Z,\C^m)$, where it can be shown to be self-adjoint. In the special case of $A\equiv-\id$ we call $W$ a discrete Schr{\"o}dinger operator with matrix-valued potential. For a survey on Jacobi operators with matrix-valued potentials we refer to \cite{Damanik2008}. Let $|\lambda_1|\ge|\lambda_2|\ge\dots\ge2$ bet the eigenvalues of $W$ outside $[-2,2]$, arranged by decreasing modulus. Each eigenvalue $\lambda_j$ has multiplicity $m_j\le m$ and $|k_j|>1$ are defined as above. 

\begin{theorem}\label{th:mainmatrix}
Assume $\tr(B^2)\in\ell^2(\Z)$ and that $\sum_{n\in\Z}\tr(A(n)A(n)^*-\id)-\log\det(A(n)A(n)^*)$ is finite. Then it holds that
\begin{align}
\sum_{j} m_j\left(k_j^2-\frac{1}{k_j^2}-\log|k_j|^4\right) 
\le\sum_{n\in\Z}\tr\big(B(n)^2\big)+2\sum_{n\in\Z}\tr\big(A(n)A(n)^*-\id\big)-\log\det\big(A(n)A(n)^*\big)
\label{eq:finalmatrix} 
\end{align}
and this inequality is sharp.
In particular, the discrete Schr{\"o}dinger operator with $A(n)=-\id$ for all $n\in\Z$ and $\tr(B^2)\in\ell^2(\Z)$ satisfies
\begin{align}
\sum_{j}  m_j\left(k_j^2-\frac{1}{k_j^2}-\log|k_j|^4\right)  \le\sum_{n\in\Z}\tr\big(B(n)^2\big)\,.
\label{eq:finalmatrix1} 
\end{align}

\end{theorem}


\section{The proof of Theorem \ref{th:main}}\label{sec:proof}
For the moment, we shall assume that the potential $b$ vanishes if only $|n|$ is large enough, that is $b(n)=0$ for $|n|>n_{\max}$. We also assume that $a(n)=-1$ for $|n|>n_{\max}$. Under these assumptions the Jacobi operator $W$ has a finite number of eigenvalues outside of $[-2,2]$, which we denote by 
\begin{align*}
\lambda_1<\dots<\lambda_L<-2<2<\lambda_{L+1}<\dots<\lambda_M\,.
\end{align*}
Here, we have slightly changed the enumeration of the eigenvalues such that the first $L$ eigenvalues are negative and the rest positive. 
 Each of the $\lambda_j$ can be written uniquely as $\lambda_j=-k_j-\frac{1}{k_j}$, where $|k_j|\ge1$. 
 
\subsection{Elimination of negative eigenvalues}
Let $\varphi$ be the eigenfunction to the lowest eigenvalue $\lambda_1$, i.e. the solution to the equation $W\varphi=\lambda_1\varphi$.
Let us assume that $\lambda_1<-2$, so that $k_1>1$.
For $n<-n_{\max}$ the eigenequation  $W\varphi=\lambda_1\varphi$ takes a very simple form, 
\begin{align*}
-\varphi(n-1)-\varphi(n+1)=\lambda_1\varphi(n)\,.
\end{align*}
It is easy to verify that a complete set of solutions of this equation is given by the functions $k_1^n$ and $k_1^{-n}$. Since $\varphi\in\ell^2(\Z)$ and $k_1>1$ we conclude that $\varphi(n)=ck_1^n$ for $n<-n_{\max}$. Similarly we can treat the case $n>n_{\max}$ and we can describe the behaviour of $\varphi$ for sufficiently large or small $n$ as
\begin{align}
\varphi(n)=\begin{cases}
c k_1^n&,\,n<-n_{\max}\\
d k_1^{-n}&,\,n>n_{\max}
\end{cases}
\label{eq:phiasymp} 
\end{align}
with constants $c,d\in\R$. These constants can be chosen in such a way that $\varphi(n)>0$, which follows from the fact that $(W-\lambda_1)\ge0$, see for example \cite[Theorem 2.8]{Gesztesy1993b}.

Following \cite{Gesztesy1996} we now introduce the operator $D$ on $\ell^2(\Z)$ as
\begin{align*}
(Du)(n)=-\sqrt{\frac{-a(n)\varphi(n)}{\varphi(n+1)}}u(n+1)+\sqrt{\frac{-a(n)\varphi(n+1)}{\varphi(n)}}u(n)
\end{align*}
and note that its adjoint is given by
\begin{align*}
(D^*u)(n)=-\sqrt{\frac{-a(n-1)\varphi(n-1)}{\varphi(n)}}u(n-1)+\sqrt{\frac{-a(n)\varphi(n+1)}{\varphi(n)}}u(n)\,.
\end{align*}
The eigenequation $W\varphi=\lambda_1\varphi$ can be written as
\begin{align}
a(n-1)\frac{\varphi(n-1)}{\varphi(n)}+a(n)\frac{\varphi(n+1)}{\varphi(n)}=\lambda_1-b(n)
\label{eq:eigeneq} 
\end{align} 
which can be used to show that $D^*D=W-\lambda_1$. This identity allows us to conclude that the spectrum of the operator $D^*D$ coincides with the spectrum of $W$, shifted by $-\lambda_1$. In particular, the eigenvalues of $D^*D$ are $0,\lambda_2-\lambda_1,\dots,\lambda_M-\lambda_1$. 
We also consider the operator $DD^*$, which can be written as $DD^*=W_1-\lambda_1$ where $W_1$ is a Jacobi operator given by
\begin{align*}
(W_1u)(n)=a_1(n-1)u(n-1)+a_1(n)u(n+1)+b_1(n)u(n)\,,
\end{align*}
with off-diagonal entries 
\begin{align}
a_1(n)=-\frac{\sqrt{a(n)a(n+1)\varphi(n)\varphi(n+2)}}{\varphi(n+1)}
\label{eq:defa1} 
\end{align}
and potential
\begin{align}
b_1(n)=-a(n)\left(\frac{\varphi(n)}{\varphi(n+1)}+\frac{\varphi(n+1)}{\varphi(n)}\right)+\lambda_1\,.
\label{eq:defV1} 
\end{align}
Note that $a_1(n)=-1$ and $b_1(n)=0$ if only $|n|$ is sufficiently large.
A general result (see e.g. \cite{Deift1978}) shows that the operators $D^*D$ and $DD^*$ have the same spectrum with the possible exception of the eigenvalue zero. Assume that there is a non-vanishing function $\psi$ such that $DD^*\psi=0$. It follows that $D^*\psi=0$ and as a consequence, for $n<-n_{\max}$, it must hold that 
\begin{align*}
-\sqrt{\frac{1}{k_1}}\psi(n-1)+\sqrt{k_1}\psi(n)=0\,.
\end{align*}
Multiplying this equation by $\sqrt{k_1}$ we find that $\psi(n)=p k_1^{-n}$ with $p\in\R$. This function is only in $\ell^2(-\N)$ if $p=0$, which implies that $\psi$ vanishes everywhere. Consequently zero is not an eigenvalue of $DD^*$. Thus the Schr{\"o}dinger operator $W_1$ has precisely the eigenvalues $\lambda_2,\dots,\lambda_M$. 
We now aim to express $\sum_{n\in\Z}b_1(n)^2$ in terms of $k_1$ and the potential $b$. To this end we recall \eqref{eq:defV1} and compute that
\begin{align}
\sum_{n\in\Z}b_1(n)^2=
\sum_{n\in\Z}a(n)^2\left(\frac{\varphi(n)}{\varphi(n+1)}+\frac{\varphi(n+1)}{\varphi(n)}\right)^2-
2\lambda_1a(n)\left(\frac{\varphi(n)}{\varphi(n+1)}+\!\frac{\varphi(n+1)}{\varphi(n)}\right)+
\lambda_1^2\,.
\label{eq:V1square} 
\end{align}
Note that  the summands vanish if $|n|$ is sufficiently large, say if $|n|>N$. This allows us to write all the involved series as finite sums  and reorder as we please.
Using \eqref{eq:eigeneq} we can compute that
\begin{align*}
\sum_{n=-N}^N a(n)\left(\frac{\varphi(n)}{\varphi(n+1)}+\frac{\varphi(n+1)}{\varphi(n)}\right)&=\sum_{n=-N+1}^{N+1}a(n-1)\frac{\varphi(n-1)}{\varphi(n)}+\sum_{n=-N}^Na(n)\frac{\varphi(n+1)}{\varphi(n)}\\
&=a(N)\frac{\varphi(N)}{\varphi(N+1)}-a(-N-1)\frac{\varphi(-N-1)}{\varphi(-N)}\\
&\phantom{=}+\sum_{n=-N}^N \big(\lambda_1-b(n)\big)
\end{align*}
The two terms outside the sum can be calculated by recalling that choosing $N$ sufficiently large guarantees that $a(N)=a(-N-1)=-1$. From \eqref{eq:phiasymp} we then obtain  
\begin{align}
\sum_{n=-N}^N& a(n)\left(\frac{\varphi(n)}{\varphi(n+1)}+\frac{\varphi(n+1)}{\varphi(n)}\right)=
-k_1+\frac{1}{k_1} +\sum_{n=-N}^N \big(\lambda_1-b(n)\big).
\label{eq:1terms} 
\end{align}
To treat the remaining terms in \eqref{eq:V1square} we consider the square of equation \eqref{eq:eigeneq}. This allows us to establish that 
\begin{align*}
\sum_{n=-N}^N a(n)^2\left(\frac{\varphi(n)^2}{\varphi(n+1)^2}+\frac{\varphi(n+1)^2}{\varphi(n)^2}\right)
&=\sum_{n=-N+1}^{N+1}a(n-1)^2\frac{\varphi(n-1)^2}{\varphi(n)^2}+\sum_{n=-N}^N a(n)^2\frac{\varphi(n+1)^2}{\varphi(n)^2}\\
&=a(N)^2\frac{\varphi(N)^2}{\varphi(N+1)^2}-a(-N-1)^2\frac{\varphi(-N-1)^2}{\varphi(-N)^2}\\
&\phantom{=}+\!\sum_{n=-N}^N\!\!\big(\lambda_1-b(n)\big)^2-2a(n)a(n-1)\frac{\varphi(n-1)\varphi(n+1)}{\varphi(n)^2}.
\end{align*}
Note that the second factor in the sum on the right-hand-side equals $a_1(n-1)^2$ and using once again \eqref{eq:phiasymp} as well as $a(N)=a(-N-1)=-1$ we obtain the identity
\begin{align}
\sum_{n=-N}^N a(n)^2\!\left(\frac{\varphi(n)^2}{\varphi(n+1)^2}+\frac{\varphi(n+1)^2}{\varphi(n)^2}\right)=
k_1^2-\frac{1}{k_1^2}+\!\sum_{n=-N}^N\big(\lambda_1-b(n)\big)^2\!-2a_1(n-1)^2\,.
\label{eq:2terms} 
\end{align}
Inserting \eqref{eq:1terms} and \eqref{eq:2terms} into \eqref{eq:V1square} finally yields
\begin{align*}
\sum_{n=-N}^N b_1(n)^2= -k_1^2+\frac{1}{k_1^2}+\sum_{n=-N}^Nb(n)^2
+2\sum_{n=-N}^N \big(a(n)^2-a_1(n-1)^2\big)\,.
\end{align*}
The summands in all three sums vanish if $n$ is sufficiently large or small. In particular we can rewrite the equation as
\begin{align*}
\sum_{n\in\Z}b_1(n)^2=-k_1^2+\frac{1}{k_1^2}+\sum_{n\in\Z}b(n)^2+2\sum_{n\in\Z}\big(a(n)^2-a_1(n)^2\big)\,.
\end{align*}
The off-diagonal entries $a_1(n)$ are linked to $a(n)$ via \eqref{eq:defa1} and we observe that
\begin{align*}
\prod_{n=-N}^N a_1(n)^2
&=\prod_{n=-N}^{N}\frac{\varphi(n)}{\varphi(n+1)}\prod_{n=-N+1}^{N+1}\frac{\varphi(n+1)}{\varphi(n)}
\prod_{n=-N}^N a(n)a(n+1)\\
&=\frac{\varphi(-N)}{\varphi(-N+1)}\frac{\varphi(N+2)}{\varphi(N+1)}\prod_{n=-N}^N a(n)a(n+1)\,.
\end{align*}
Recalling \eqref{eq:phiasymp} and that $a(n)=a_1(n)=-1$ for sufficiently large $|n|$, this shows that
\begin{align*}
\prod_{n\in\Z}a_1(n)^2=\frac{1}{k_1^2}\prod_{n\in\Z}a(n)^2\,.
\end{align*}

If $\lambda_2<-2$ we can now repeat the above procedure for the operator $W_1$. Starting from the eigenfunction $\varphi$ for the lowest eigenvalue $\lambda_2$ of $W_1$ we can construct a new operator $W_2$ with off-diagonal entries $a_2$ and with potential $b_2$, which are defined as in \eqref{eq:defa1} and \eqref{eq:defV1}, respectively, where $a_1$ now takes the place of $a$ and $\lambda_2$ takes the place of $\lambda_1$. The operator $W_2$ then has the eigenvalues $\lambda_3,\dots,\lambda_M$ and it holds that
\begin{align*}
\sum_{n\in\Z}b_2(n)^2=-k_2^2+\frac{1}{k_2^2}+\sum_{n\in\Z}b_1(n)^2+2\sum_{n\in\Z}\big(a_1(n)^2-a_2(n)^2\big)
\end{align*}
as well as
\begin{align*}
\prod_{n\in\Z}a_2(n)^2=\frac{1}{k_2^2}\prod_{n\in\Z}a_1(n)^2\,.
\end{align*}
We can proceed in this manner eliminating all the negative eigenvalues and thus prove the following proposition.

\begin{proposition}\label{prop:positive}
Let $W$ be the Jacobi operator \eqref{eq:defH} with $a(n)=-1$ and $b(n)=0$ if only $|n|$ is sufficiently large. 
Let $\lambda_1<\dots<\lambda_L<-2$ be the negative eigenvalues of $W$ and define $k_j>1 $ such that $\lambda_j=-k_j-\frac{1}{k_j}$. It holds that
\begin{align}
\sum_{n\in\Z}b_L(n)^2=\sum_{n\in\Z}b(n)^2-\sum_{j=1}^L\left( k_j^2-\frac{1}{k_j^2}\right)+2\sum_{n\in\Z}\big(a(n)^2-a_L(n)^2\big)
\label{eq:pos1} 
\end{align}
where the off-diagonal entries $a_L(n)$ and the potential $b_L(n)$ are obtained recursively by the process described in detail above. It also holds that
\begin{align}
\prod_{n\in\Z}a_L(n)^2=\prod_{j=1}^L\frac{1}{k_j^2}\,\prod_{n\in\Z}a(n)^2\,.
\label{eq:pos2} 
\end{align}

\end{proposition}

The Jacobi operator $(W_Lu)(n)=a_L(n-1)u(n-1)+a_L(n)u(n+1)+b_L(n)u(n)$ has only positive eigenvalues $2<\lambda_{L+1}<\dots<\lambda_M$ and the property that $a_L(n)=-1$ and $b_L(n)=0$ if only $|n|$ is sufficiently large.

\subsection{Elimination of positive eigenvalues}
To eliminate the positive eigenvalues we have to proceed in a slightly different manner. The reason is that $W_L-\lambda_{L+1}$ is not a positive operator. This  means that $\varphi$, the eigenfunction corresponding to $\lambda_{L+1}$, is not strictly positive which can also be seen from the fact that  $k_{L+1}<-1$ and that 
$\varphi(n)=ck_{L+1}^n$ for sufficiently small $n$.
To circumvent this problem, we consider the new operator $W_L'$ which is defined as
\begin{align*}
(W_L'u)(n)=a_L(n-1)u(n-1)+a_L(n)u(n+1)-b_L(n)u(n)\,.
\end{align*}
If $\psi\in\ell^2(\Z)$ is an eigenfunction of $W_L$ with eigenvalue $\lambda$ then $\psi'(n)=(-1)^{n}\psi(n)$ is also a sequence in $\ell^2(\Z)$ and we compute that
\begin{align*}
(W_L'\psi')(n)
=-(-1)^{n}(W_L\psi)(n)=-\lambda\psi'(n)\,.
\end{align*}
Thus $-\lambda$ is an eigenvalue of $W_L'$ and a similar computation shows that every eigenvalue $\eta$ of $W_L'$ corresponds to an eigenvalue $\lambda=-\eta$ of the operator $W_L$. We conclude that the eigenvalues of $W_L'$ are exactly $-\lambda_M<\dots<-\lambda_{L+1}<-2$.  These eigenvalues can be written as $\eta_j=-\lambda_j=-\ell_j-\frac{1}{\ell_j}$ with $\ell_j=-k_j>1$. 
The results of Proposition \ref{prop:positive} now also hold for the operator $W_L'$ and we thus obtain the identity
\begin{align}
\sum_{n\in\Z}b_{M}(n)^2=\sum_{n\in\Z}b_L(n)^2-\sum_{j=L+1}^M \left(\ell_{j}^2-\frac{1}{\ell_{j}^2}\right)+2\sum_{n\in\Z}\big(a_L(n)^2-a_{M}(n)^2\big)
\label{eq:neg1} 
\end{align}
as well as
\begin{align}
\prod_{n\in\Z}a_M(n)^2=\prod_{j=L+1}^M\frac{1}{\ell_j^2}\,\prod_{n\in\Z}a_L(n)^2\,.
\label{eq:neg2} 
\end{align}
The Jacobi operator $W_M$ with potential $b_M$ and off-diagonal entries $a_M$ emerges from the (all-together) $M$-th elimination process in which we remove the last remaining eigenvalue $\eta_{L+1}=-\lambda_{L+1}$ from the spectrum and $W_M$ has thus no eigenvalues outside of $[-2,2]$.
Recalling that $\ell_j=-k_j$ we can combine \eqref{eq:neg1} and \eqref{eq:neg2} with \eqref{eq:pos1} and \eqref{eq:pos2} to obtain
\begin{align}
\sum_{n\in\Z}b_M(n)^2=\sum_{n\in\Z}b(n)^2-\sum_{j=1}^M\left( k_j^2-\frac{1}{k_j^2}\right)+2\sum_{n\in\Z}\big(a(n)^2-a_M(n)^2\big)
\label{eq:final1} 
\end{align}
and also
\begin{align}
\prod_{n\in\Z}a_M(n)^2=\prod_{j=1}^M\frac{1}{k_j^2}\,\prod_{n\in\Z}a(n)^2\,.
\label{eq:aNa} 
\end{align}
The left-hand-side of \eqref{eq:final1} is clearly positive so that we can deduce the inequality
\begin{align}
\sum_{j=1}^M\left( k_j^2-\frac{1}{k_j^2}\right)\le\sum_{n\in\Z}b(n)^2+2\sum_{n\in\Z}(a(n)^2-a_M(n)^2)\,.
\label{eq:final2} 
\end{align}
We now aim to find an upper bound on the term $\sum_{n\in\Z}\big(a(n)^2-a_M(n)^2\big)$ that only explicitly  depends on $a(n)$ and the eigenvalues of $W$. To this end we apply the logarithmic function to both sides of \eqref{eq:aNa} and derive that
\begin{align*}
\sum_{n\in\Z}\left( \log a(n)^2-\log a_M(n)^2\right)=\sum_{j=1}^M\log |k_j|^2\,.
\end{align*}
For positive real numbers $x\in(0,\infty)$ it holds that $\log x\le x-1$ and using this inequality and the identity above we can conclude that
\begin{align*}
\sum_{n\in\Z} \left(a(n)^2-a_M(n)^2\right) \le\sum_{j=1}^M\log |k_j|^2+\sum_{n\in\Z} \left(a(n)^2-1-\log a(n)^2\right)\,.
\end{align*}
Inserting this inequality into \eqref{eq:final2} yields the desired result
\begin{align*}
\sum_{j=1}^M\left( k_j^2-\frac{1}{k_j^2}-\log |k_j|^4\right)\le\sum_{n\in\Z}b(n)^2+2\sum_{n\in\Z}\left( a(n)^2-1-\log a(n)^2\right)\,.
\end{align*}

It remains to extend this result to Jacobi operators which do not necessarily satisfy that $b(n)=0$ and $a(n)=-1$ for sufficiently large $|n|$. For a given Jacobi operator $W$ and $N\in\N\cup\{0\}$ we define the truncation $W_N$ as the Jacobi operator with potential and off-diagonal terms specified by 
\begin{align*}
b_N(n)=\begin{cases}
b(n)&,\,|n|< N\\
0&,\,|n|\ge N\,,
\end{cases}
&&
a_N(n)=\begin{cases}
a(n)&,\,|n|< N\\
-1&,\,|n|\ge N\,.
\end{cases}
\end{align*}
In particular, for $N=0$ we obtain the free discrete Schr{\"o}dinger operator $W_0$ with $b_0(n)=0$ and $a_0(n)=-1$ for all $n\in\Z$. Each $W_N$ fulfils the initial assumptions of our proof and thus satisfies inequality \eqref{eq:final}. 
If $W$ is such that the right-hand-side of \eqref{eq:final} is finite, then $W-W_0$ is compact and a perturbation argument (see e.g. \cite[Theorem 6.2]{Killip2003}) shows that the left-hand-side of \eqref{eq:final} for $W_N$ converges to the left-hand-side of \eqref{eq:final} for $W$ as $N\to\infty$. Similarly the right-hand-sides converge and thus the inequality also holds for $W$.


\section{The proof of Theorem \ref{th:mainmatrix}}

For now we also assume that $B$ is compactly supported, so that the potential vanishes if $|n|>n_{\max}$ and that also $A(n)=-\id$ for $|n|>n_{\max}$.
These assumptions will later be dropped by a continuity argument but for the moment they guarantee that $W$  has only finitely many eigenvalues outside of $[-2,2]$. 
Let these eigenvalues be denoted by
\begin{align*}
\lambda_1<\dots<\lambda_L<-2<2<\lambda_{L+1}<\dots<\lambda_M\,.
\end{align*}
Each eigenvalue can be written as $\lambda_j=-k_j-\frac{1}{k_j}$ with $|k_j|\ge1$ and has a multiplicity $m_j\le m$.
We assume that $\lambda_1<-2$ and thus $k_1>1$. Define $\Phi$ to be the matrix solution of the equation
\begin{align*}
A(n-1)^*\Phi(n-1)+A(n)\Phi(n+1)+B(n)\Phi(n)=\lambda_1\Phi(n)
\end{align*}
which satisfies
\begin{align}
\Phi(n)=k_1^n\id
\label{eq:Phiasymp} 
\end{align}
for $n<-n_{\max}$. 
Note that the quadratic form of $W$ is given by 
\begin{align}
\sklp[\ell^2]{Wu,u}=\sum_{n\in\Z}\sklp[\C^m]{A(n-1)^*u(n-1)+A(n)u(n+1)+B(n)u(n),u(n)}\,.
\label{eq:quadraticform} 
\end{align}

\begin{lemma}\label{lem:Phi}
$\Phi(n)$ is invertible for all $n\in\Z$. Furthermore it holds that $-A(n)\Phi(n+1)\Phi(n)^{-1}$ is an invertible, Hermitian, positive-semidefinite matrix for all $n\in\Z$. 
\end{lemma}

\begin{proof}
The proof of invertibility is adapted from similar results \cite{Benguria2000} in the case of a Schr{\"o}dinger operator on $L^2(\R)$.
The matrix function $\Phi(n)$ is by construction invertible for all $n<-n_{\max}$. Assume that there exists an $N\in\N$ and a $u\in\C^m$ such that $\Phi(N)u=0$. Clearly the vector function $\Phi(n)u$ satisfies $(W-\lambda_1)\Phi(n)u=0$. We can now define a new vector function $\varphi$ as
\begin{align*}
\varphi(n)=\begin{cases}
\Phi(n)u&,\,n<N\\
0&,\,n\ge N
\end{cases}
\end{align*}
and from \eqref{eq:quadraticform} we see that $\sklp[\ell^2]{W\varphi,\varphi}=\lambda_1\sklp[\ell^2]{\varphi,\varphi}$. As $\lambda_1$ is the ground state we conclude that $\varphi$ must therefore satisfy $(W-\lambda_1)\varphi=0$. However, this implies that $\varphi(n)=0$ for all $n\in\N$ which is a contradiction as $\Phi(n)$ is invertible for $n<-n_{\max}$. 

We proceed similarly as in the scalar case in \cite{Gesztesy1996} and note that  $(W-\lambda_1)\ge 0$. For any $N\in\Z$ consider now the operator $W_N$ which is defined as $W$ restricted to the space $\ell^2(\Z\cap(-\infty,N),\C^m)$ with a Dirichlet boundary condition at $N$, i.e.
\begin{align*}
(W_Nu)(n)=\begin{cases}
(Wu)(n)&,\,n\le N-2\\
A(N-2)^*u(n-2)+B(N-1)u(N-1)&,\,n=N-1\,.
\end{cases}
\end{align*}
Then, by the Min-Max principle, $(W_N-\lambda_1)\ge0$. Let $v\in\C^m$ be an arbitrary vector and consider the vector function $\varphi(n)=\Phi(n)\Phi(N)^{-1}A(N-1)^{-1}v$. By construction it holds that $(W-\lambda_1)\varphi=0$. We define the vector function $\psi\in\ell^2(\Z\cap(-\infty,N),\C^m)$ as $\psi(n)=\varphi(n)$ for $n\le N-1$.
For $n\le N-2$ it holds that $\big((W_N-\lambda_1)\psi\big)(n)=0$ and for $n=N-1$ we calculate 
\begin{align*}
\big((W_N-\lambda_1)\psi\big)(N-1)\!=\!A(N-2)^*\psi(N-2)+\!\big(B(N-1)-\lambda_1\big)\psi(N-1)\!
=\!-A(N-1)\varphi(N)\!=\!-v
\end{align*}
and thus conclude that $(W_N-\lambda_1)\psi=-\delta_{N-1}v$.
As $(W_N-\lambda_1)\ge0$ it holds that
\begin{align*}
0&\le\sklp[\ell^2]{(W_N-\lambda_1)\psi,\psi}=-\sklp[\ell^2]{\delta_{N-1}v,\psi}
=-\sklp[\C^m]{v,\psi(N-1)}\\
&=-\sklp[\C^m]{v,\Phi(N-1)\Phi(N)^{-1}A(N-1)^{-1}v}
\end{align*} 
This shows that $-A(N-1)\Phi(N)\Phi(N-1)^{-1}$ is a positive-semidefinite matrix and in particular Hermitian. 
\end{proof}

Since $(W-\lambda_1)\Phi=0$, we can conclude that
\begin{align*}
A(n-1)^*\Phi(n-1)\Phi(n)^{-1}+A(n)\Phi(n+1)\Phi(n)^{-1}=\lambda_1-B(n)\,.
\end{align*}
Defining $F(n)=-A(n)\Phi(n+1)\Phi(n)^{-1}$,  this can be written as a discrete Riccati-type equation
\begin{align}
A(n-1)^*F(n-1)^{-1}A(n-1)+F(n)=B(n)-\lambda_1\,.
\label{eq:riccatimatrix} 
\end{align}
We now introduce the operator $D$ on $\ell^2(\Z,\C^m)$ as
\begin{align*}
(Du)(n)=F(n)^{-\frac12}A(n)u(n+1)+F(n)^{\frac12}u(n)
\end{align*}
where we note that by Lemma \ref{lem:Phi} the matrices $F(n)^{\pm1/2}$ are well-defined.
The adjoint of $D$ is given by
\begin{align*}
(D^*u)(n)=A(n-1)^*F(n-1)^{-\frac12}u(n-1)+F(n)^{\frac12}u(n)\,.
\end{align*}
Using \eqref{eq:riccatimatrix} it can be shown that $D^*D=W-\lambda_1$. This identity allows us to conclude that the spectrum of the operator $D^*D$ coincides with the spectrum of $W$, shifted by $-\lambda_1$. In particular, the eigenvalues of $D^*D$ are $0,\lambda_2-\lambda_1,\dots,\lambda_M-\lambda_1$ with multiplicities $m_1,\dots, m_M$. 
We also consider the operator $DD^*$, which can be written as $DD^*=W_1-\lambda_1$ where $W_1$ is the Jacobi operator given by
\begin{align*}
(W_1u)(n)=A_1(n-1)^*u(n-1)+A_1(n)u(n+1)+B_1(n)u(n)\,,
\end{align*}
with off-diagonal entries 
\begin{align}
A_1(n)=F(n)^{-\frac12}A(n)F(n+1)^{\frac12}
\label{eq:defAmatrix} 
\end{align}
and potential
\begin{align}
B_1(n)=F(n)^{-\frac12}A(n)A(n)^*F(n)^{-\frac12}+F(n)+\lambda_1\,.
\label{eq:defVmatrix} 
\end{align}
Note that $A_1(n)=-\id$ and $B_1(n)=0$ if only $n$ is sufficiently small.
A general result (see e.g. \cite{Deift1978}) shows that, with the possible exception of the eigenvalue zero, the operators $D^*D$ and $DD^*$ have the same eigenvalues with the same multiplicities. Assume that there is a non-vanishing function $\psi$ such that $DD^*\psi=0$. It follows that $D^*\psi=0$ and as a consequence, for $n<-n_{\max}$, it must hold that 
\begin{align*}
-\sqrt{\frac{1}{k_1}}\psi(n-1)+\sqrt{k_1}\psi(n)=0\,.
\end{align*}
Multiplying this equation by $\sqrt{k_1}$ we find that $\psi(n)=k_1^{-n}v$ with $v\in\C^m$. This function is clearly only in $\ell^2(-\N, \C^m)$ if $v=0$, which implies that $\psi$ vanishes everywhere. Consequently zero is not an eigenvalue of $DD^*$. Thus the Schr{\"o}dinger operator $W_1$ has precisely the eigenvalues $\lambda_2,\dots,\lambda_M$ with multiplicities $m_2,\dots,m_M$. 

From the construction of $F$ it immediately follows that $F(n)=k_1\id$ for $n<-n_{\max}$ and thus
\begin{align*}
\tr\big(F(n)\big)=mk_1,&&
&\tr\big(F(n)^{-1}\big)=m\frac{1}{k_1}
\end{align*}
as well as
\begin{align*}
\tr\big(F(n)^2\big)=mk_1^2,&&
\tr\big(F(n)^{-2}\big)=m\frac{1}{k_1^2}\,.
\end{align*}
The next result describes the behaviour of all these terms if $n\to\infty$.

\begin{lemma}\label{lem:F}
At any point $n_0>n_{\max}$ the eigenvectors of $F(n_0)$ do not depend on $n_0$. The matrix $F(n_0)$ has an eigenvalue $\frac{1}{k_1}$ of multiplicity $m_1$ and every other eigenvalue $\mu(n_0)$ of $F(n_0)$ gives rise to a branch of eigenvalues $\mu(n)$ of $F(n)$ where $\mu(n)$ converges exponentially fast to $k_1$ for $n\to\infty$. 
In particular it holds that
\begin{align}
\lim_{n\to\infty}\tr\big(F(n)\big)=m_1\frac{1}{k_1}+(m-m_1)k_1,&&
&\lim_{n\to\infty}\tr\big(F(n)^{-1}\big)=m_1k_1+(m-m_1)\frac{1}{k_1}
\label{eq:asympF1} 
\end{align}
as well as
\begin{align}
\lim_{n\to\infty}\tr\big(F(n)^2\big)=m_1\frac{1}{k_1^2}+(m-m_1)k_1^2,&&
\lim_{n\to\infty}\tr\big(F(n)^{-2}\big)=m_1k_1^2+(m-m_1)\frac{1}{k_1^2}
\label{eq:asympF2} 
\end{align}
where the limits are achieved exponentially fast. Furthermore the potential $B_1(n)$ converges to $0$ and the off-diagonal entries $A_1(n)$ converge to $-\id$ exponentially fast for $n\to\infty$. 
\end{lemma}

\begin{proof}
If $n>n_0>n_{\max}$ then $\Phi(n)=Pk_1^{n-n_0}+Qk_1^{n_0-n}$ with matrices $P,Q\in\C^{m\times m}$. Since $(W-\lambda_1)\Phi=0$ it must hold that 
\begin{align*}
-\Phi(n_0)-\Phi(n_0+2)&=\lambda_1\Phi(n_0+1)\\
-\Phi(n_0+1)-\Phi(n_0+3)&=\lambda_1\Phi(n_0+2)
\end{align*}
and inserting $\Phi(n)=Pk_1^{n-n_0}+Qk_1^{n_0-n}$ into these two equations allows us to compute the matrices $P$ and $Q$ as
\begin{align*}
P&=\frac{k_1}{k_1^2-1}\left(-\frac{1}{k_1}\Phi(n_0)+\Phi(n_0+1)\right)\\
Q&=\frac{k_1}{k_1^2-1}\big(k_1\Phi(n_0)-\Phi(n_0+1)\big).
\end{align*}
Thus we can conclude that
\begin{align*}
\Phi(n)=
\frac{k_1}{k_1^2-1}\left(\left(k_1^{n_0-n+1}-k_1^{n-n_0-1}\right)\Phi(n_0)+\left(k_1^{n-n_0}-k_1^{n_0-n}\right)\Phi(n_0+1)\right)
\end{align*}
and as a consequence 
\begin{align}
\Phi(n)=
\frac{k_1}{k_1^2-1}\left(k_1^{n_0-n+1}-k_1^{n-n_0-1}+\left(k_1^{n-n_0}-k_1^{n_0-n}\right)F(n_0)\right)\Phi(n_0)\,.
\label{eq:Phinn0} 
\end{align}
Every eigenfunction $\varphi$ of $W$ to the eigenvalue $\lambda_1$ can be written as $\varphi(n)=\Phi(n)u$ with $u\in\C^m$. In order for $\varphi$ to be square-summable, it is necessary that the exponentially growing parts in \eqref{eq:Phinn0} vanish. As a consequence it must hold that $\varphi(n)=\Phi(n)\Phi(n_0)^{-1}v$ with $v\in\C^m$ such that $F(n_0)v=\frac{1}{k_1}v$. Since $\lambda_1$ is an $m_1$-fold eigenvalue of $W$ we conclude that $F(n_0)$ has an $m_1$-fold eigenvalue $\frac{1}{k_1}$. 
From \eqref{eq:Phinn0} we also conclude that $F(n)=f_n\big(F(n_0)\big)$ with
\begin{align*}
f_n(z)=\frac{k_1^{2n_0-2n}-1+\left(k_1-k_1^{2n_0-2n-1}\right)z}
{k_1^{2n_0-2n+1}-\frac{1}{k_1}+\left(1-k_1^{2n_0-2n}\right)z}\,.
\end{align*}
This shows that the eigenvalues of $F(n)$ are independent of $n$ and that every eigenvalue $\mu$ of $F(n_0)$ generates an eigenvalue $\mu(n)=f_n(\mu)$ for $F(n)$. We see that $f_n\left(\frac{1}{k_1}\right)=\frac{1}{k_1}$ and $\lim_{n\to\infty}f_n(\mu)=k_1$ for all other $\mu$. 

As the eigenvectors of $F(n)$ do not depend on $n$, we can find a transformation $T$ such that $F(n)=TE(n)T^{-1}$. The matrix $E(n)$ is diagonal with the first $m_1$ entries being $\frac{1}{k_1}$ followed by the other eigenvalues of $F(n)$. Using the definition \eqref{eq:defVmatrix} we see that $B_1(n)=T(E(n)^{-1}+E(n)+\lambda_1)T^{-1}$ which converges to 0 for $n\to\infty$ by the results from above and the definition of $k_1$. Similarly it follows that $\lim_{n\to\infty}A(n)=-\lim_{n\to\infty}TE(n)^{-\frac12}E(n+1)^{\frac12}T^{-1}=-\id$. 
This finishes the proof.
\end{proof}

We now aim to express $\sum_{n\in\Z}\tr\big(B_1(n)^2\big)$ in terms of $k_1$ and the potential $B$. To this end we recall \eqref{eq:defVmatrix} and compute that
\begin{align}
\begin{split}
\sum_{n\in\Z}\tr\big(B_1(n)^2\big)=
\sum_{n\in\Z}&\tr\left(\big(F(n)^{-\frac12}A(n)A(n)^*F(n)^{-\frac12}+F(n)\big)^2\right)\\
&+2\lambda_1\tr\left(F(n)^{-\frac12}A(n)A(n)^*F(n)^{-\frac12}+F(n)\right)+m\lambda_1^2\,.
\end{split}
\label{eq:V1squarematrix} 
\end{align}
Using \eqref{eq:riccatimatrix} and the cyclicity of the trace we can deduce that for any natural number $N$ the identity
\begin{align*}
&\sum_{n=-N}^N\tr\left(F(n)^{-\frac12}A(n)A(n)^*F(n)^{-\frac12}+F(n)\right)\\
&=\sum_{n=-N+1}^{N+1}\tr\big(A(n-1)^*F(n-1)^{-1}A(n-1)\big)+\sum_{n=-N}^{N}\tr\big(F(n)\big)\\
&=\tr\big(A(N)^*F(N)^{-1}A(N)\big)\!-\!\tr\big(A(-N-1)^*F(-N-1)^{-1}A(-N-1)\big)\!+\!\!\sum_{n=-N}^N\!\!\tr\big(B(n)-\lambda_1\big)
\end{align*}
holds.
We can compute the two terms outside the sum by choosing $N$ sufficiently large such that $A(N)=A(N-1)=-\id$. From \eqref{eq:Phiasymp} we then obtain  
\begin{align}
\sum_{n=-N}^N\!\!\tr\left(F(n)^{-\frac12}A(n)A(n)^*F(n)^{-\frac12}+F(n)\right)
\!=\tr\big(F(N)^{-1}\big)-m\frac{1}{k_1}+\!\sum_{n=-N}^N\!\!\tr\big(B(n)-\lambda_1\big).
\label{eq:1termsmatrix} 
\end{align}
It remains to treat the quadratic terms in \eqref{eq:V1squarematrix} and to this end we consider the square of equation \eqref{eq:riccatimatrix}. This allows us together with the cyclicity of the trace to conclude that
\begin{align*}
&\sum_{n=-N}^N \tr\left(\big(F(n)^{-\frac12}A(n)A(n)^*F(n)^{-\frac12}+F(n)\big)^2\right)\\
&=\sum_{n=-N}^N\tr\left(\big(A(n)^*F(n)^{-1}A(n)\big)^2\right)+2\tr\big(A(n)A(n)^*\big)+\tr\left(F(n)^2\right)\\
&=\sum_{n=-N+1}^{N+1}\tr\left(\big(A(n-1)^*F(n-1)^{-1}A(n-1)\big)^2\right)+\sum_{n=-N}^N\tr\left(F(n)^2\right)+2\tr\big(A(n)A(n)^*\big)\\
&=
\tr\left(\big(A(N)^*F(N)^{-1}A(N)\big)^2\right)-\tr\left(\big(A(-N-1)^*F(-N-1)^{-1}A(-N-1)\big)^2\right)\\
&\phantom{=}+\sum_{n=-N}^N\tr\left(\big(B(n)-\lambda_1\big)^2\right)-2\tr\big(A(n-1)^*F(n-1)^{-1}A(n-1)F(n)\big)+2\tr\big(A(n)A(n)^*\big).
\end{align*}
Using once again that $A(N)=A(-N-1)=-\id$  and \eqref{eq:Phiasymp} as well as \eqref{eq:defAmatrix} we arrive at
\begin{align*}
&\sum_{n=-N}^N \tr\left(\big(F(n)^{-\frac12}A(n)A(n)^*F(n)^{-\frac12}+F(n)\big)^2\right)\\
&=\tr\big(F(N)^{-2}\big)-m\frac{1}{k_1^2}+\sum_{n=-N}^N\!\tr\left(\big(B(n)-\lambda_1\big)^2\right)\!-2\tr\big(A_1(n-1)A_1(n-1)^*\big)\!+2\tr\big(A(n)A(n)^*\big).
\end{align*}
Together with \eqref{eq:V1squarematrix} and \eqref{eq:1termsmatrix} we obtain the following identity
\begin{align*}
\sum_{n=-N}^N\tr\big(B_1(n)^2\big)
=&\tr\big(F(N)^{-2}\big)-m\frac{1}{k_1^2}+2\lambda_1\left(\tr\big(F(N)^{-1}\big)-m\frac{1}{k_1}\right)\\
&+\sum_{n=-N}^N\tr\big(B(n)^2\big)
+2\sum_{n=-N}^N \tr\big(A(n)A(n)^*-A_1(n-1)A_1(n-1)^*\big).
\end{align*}
Recalling \eqref{eq:asympF1} and \eqref{eq:asympF2} as well as $\lambda_1=-k_1-\frac{1}{k_1}$ yields
\begin{align*}
\lim_{N\to\infty}\left(\tr\big(F(N)^{-2}\big)-m\frac{1}{k_1^2}+2\lambda_1\left(\tr\big(F(N)^{-1}\big)-m\frac{1}{k_1}\right)\right)=m_1\frac{1}{k_1^2}-m_1k_1^2
\end{align*}
and consequently we finally arrive at
\begin{align*}
\sum_{n\in\Z}\tr\big(B_1(n)^2\big)=m_1\frac{1}{k_1^2} -m_1k_1^2
+\sum_{n\in\Z}\tr\big(B(n)^2\big)
+2\sum_{n\in\Z} \tr\big(A(n)A(n)^*-A_1(n-1)A_1(n-1)^*\big).
\end{align*}
Since $\lim_{|n|\to\infty}A_1(n)=-\id$, we can replace $A_1(n-1)A_1(n-1)^*$ in the above identity by $A_1(n)A_1(n)^*$.
The off-diagonal entries $A_1(n)$ are linked to $A(n)$ via \eqref{eq:defAmatrix}. 
Thus we can compute that
\begin{align*}
\sum_{n=-N}^N\log\det\big(A_1(n)A_1(n)^*\big)
&=\sum_{n=-N}^N\log\det(A(n)F(n+1)A(n)^*F(n)^{-1}\big)\\
&=\sum_{n=-N}^N \log\det F(n+1)-\log\det F(n)+\log\det\big(A(n)A(n)^*\big)\\
&=\log\det F(N+1)-\log\det F(-N)+\sum_{n=-N}^N\log\det\big(A(n)A(n)^*\big).
\end{align*}
Letting $N\to\infty$ we obtain from \eqref{eq:Phiasymp} and Lemma \ref{lem:F} that
\begin{align*}
\sum_{n\in\Z}\log\det\big(A_1(n)A_1(n)^*\big)=-m_1\log|k_1|^2+\sum_{n\in\Z}\log\det\big(A(n)A(n)^*\big).
\end{align*}
We summarise our results in the following proposition.

\begin{proposition}\label{prop:matrix}
Let $W$ be the Jacobi operator \eqref{eq:defHmatrix} with $B(n)=0$ and $A(n)=-\id$ for sufficiently large $|n|$. Let $\lambda_1<\dots<\lambda_L<-2<2<\lambda_{L+1}<\dots<\lambda_M$ be the eigenvalues of $W$ each with multiplicity $m_j\le m$ and define $|k_j|>1 $ such that $\lambda_j=-k_j-\frac{1}{k_j}$. Then it holds that
\begin{align*}
\sum_{j=1}^M m_j\left(\frac{1}{k_j^2}-k_j^2+\log|k_j|^4\right) &+\sum_{n\in\Z}\tr\big(B(n)^2\big)-m_1\log|k_1|^4\\
&+2\sum_{n\in\Z} \tr\big(A(n)A(n)^*-A_1(n)A_1(n)^*\big) \\
&=\sum_{j=2}^M m_j\left(\frac{1}{k_j^2}-k_j^2+\log|k_j|^4\right)+\sum_{n\in\Z}\tr\big(B_1(n)^2\big)
\end{align*}
where the off-diagonal entries $A_1(n)$ and the potential $B_1(n)$ are defined as in \eqref{eq:defAmatrix} and \eqref{eq:defVmatrix}, respectively. It also holds that
\begin{align*}
\sum_{n\in\Z}\log\det\big(A_1(n)A_1(n)^*\big)=-m_1\log|k_1|^2+\sum_{n\in\Z}\log\det\big(A(n)A(n)^*\big)\,.
\end{align*}

\end{proposition}

Unfortunately we cannot repeat this procedure immediately as $B_1(n)$ is decaying exponentially fast but it does not necessarily vanish for sufficiently large $n$. Similarly $A_1(n)$ converges exponentially fast to $-\id$ but $A(n)=-\id$ may not hold for sufficiently large $n$. Thus we have to introduce a cut after which we set $A_1(n)$ to be $-\id$ and $B_1(n)$ to be zero.  With fixed $N^c>n_{\max}$ let the operator $W_1^c$ be defined as
\begin{align*}
(W_1^cu)(n)=A_1^c(n-1)^*u(n-1)+A_1^c(n)u(n)+B_1^c(n)u(n)
\end{align*}
where
\begin{align*}
A_1^c(n)= \begin{cases}
A_1(n)&,\, n<N^c\\
-\id&,\,n\ge N^c,
\end{cases}
&&
B_1^c(n)= \begin{cases}
B_1(n)&,\, n<N^c\\
0&,\,n\ge N^c.
\end{cases}
\end{align*}
The new Jacobi operator $W_1^c$ satisfies all the conditions of Proposition \ref{prop:matrix}. It has eigenvalues $\mu_j=-\ell_j^2-\frac{1}{\ell_j^2}$ with $|\ell_j|>1$ for $j\in\{2,\dots,M^c\}$ and corresponding multiplicities $n_j$ which are not necessarily the same as the eigenvalues and multiplicities of $W_1$. However, the first $m_2+\dots+m_M$ eigenvalues (counted with multiplicity), say $\mu_2,\dots,\mu_{M'}$ with $n_2+\dots+n_{M'}=m_2+\dots+m_M$, can be made arbitrarily close to the original eigenvalues $\lambda_2,\dots,\lambda_M$ (counted with multiplicity) by choosing $N^c$ sufficiently large. The potential additional eigenvalues $\mu_{M'+1},\dots,\mu_{M^c}$ can be made as close to the essential spectrum as we please. 
Proposition \ref{prop:matrix} now implies that
\begin{align*}
\sum_{j=2}^{M^c} n_j\left(\frac{1}{\ell_j^2}-\ell_j^2+\log|\ell_j|^4\right) &+\sum_{n\in\Z}\tr\big(B_1^c(n)^2\big)-n_2\log|\ell_2|^4\\
&+2\sum_{n\in\Z} \tr\big(A_1^c(n)A_1^c(n)^*-A_2(n)A_2(n)^*\big) \\
&=\sum_{j=3}^{M^c} n_j\left(\frac{1}{\ell_j^2}-\ell_j^2+\log|\ell_j|^4\right)+\sum_{n\in\Z}\tr\big(B_2(n)^2\big)
\end{align*}
and also
\begin{align*}
\sum_{n\in\Z}\log\det\big(A_2(n)A_2(n)^*\big)=-n_2\log|\ell_1|^2-\sum_{n\in\Z}\log\det\big(A_1^c(n)A_1^c(n)^*\big)\,.
\end{align*}
As a consequence we obtain that
\begin{align*}
\sum_{j=1}^M m_j\left(\frac{1}{k_j^2}-k_j^2+\log|k_j|^4\right) &+\sum_{n\in\Z}\tr\big(B(n)^2\big)
-m_1\log|k_1|^4-n_2\log|\ell_2|^4\\
&+2\sum_{n\in\Z} \tr\big(A(n)A(n)^*-A_2(n)A_2(n)^*\big) \\
&=e_2+\sum_{j=3}^{M'} n_j\left(\frac{1}{\ell_j^2}-\ell_j^2+\log|\ell_j|^4\right)+\sum_{n\in\Z}\tr\big(B_2(n)^2\big)
\end{align*} 
with an error $e_1$ stemming from replacing the matrix functions $B_1^c$ and $A_1^c$ by $B_1$ and $A_1$, respectively, and substituting $M'$ for $M^c$. Note that the total number of eigenvalues (counted with multiplicities) over which we sum on the right-hand-side is $m_2+\dots+m_M-n_2$ and thus strictly smaller than $m_2+\dots+m_M$.
It also holds that
\begin{align*}
\sum_{n\in\Z}\log\det\big(A_2(n)A_2(n)^*\big)=&f_2-m_1\log|k_1|^2-n_2\log|\ell_2|^2
+\sum_{n\in\Z}\log\det\big(A(n)A(n)^*\big)
\end{align*}
with an error $f_2$ correcting the replacement of $A_1^c$ by $A_2$. We can now introduce a cut in $A_2$ and $B_2$ and repeat the above steps to eliminate all the positive eigenvalues. For the negative eigenvalues we proceed as in the scalar case, using the unitary transformation $(Uu)(n)=(-1)^{n}u(n)$.  
After a finite number of $s\le m_1+\dots+m_M$ steps, we arrive at  
\begin{align}
\begin{split}
\sum_{j=1}^M m_j\!\left(\frac{1}{k_j^2}-k_j^2+\log|k_j|^4\right) &\!+\sum_{n\in\Z}\tr\big(B(n)^2\big)\!-m_1\log|k_1|^4\!-n_2\log|\ell_2|^4\!-\dots-q_s\log|p_s|^4\\
&+2\sum_{n\in\Z} \tr\big(A(n)A(n)^*-A_s(n)A_s(n)^*\big)\\
&=e_2+\dots +e_s+\sum_{n\in\Z}\tr\big(B_s(n)^2\big)\\
&\ge e_2+\dots+e_s
\end{split}
\label{eq:finalmatrix2} 
\end{align}
with
\begin{align}
\begin{split}
\sum_{n\in\Z}\log\det\big(A_s(n)A_s(n)^*\big)=
f_2+\dots+f_s&-m_1\log|k_1|^2-n_2\log|\ell_2|^2-\dots-q_s\log|p_s|^2\\
&+\sum_{n\in\Z}\log\det\big(A(n)A(n)^*\big).
\end{split}
\label{eq:aNamatrix} 
\end{align}
Fix $n\in\N$ and note that $A_s(n)A_s(n)^*$ is Hermitian, positive semi-definite and invertible. Let $\mu_1,\dots,\mu_m$ denote the eigenvalues of this matrix, which are necessarily real and positive. Then, using $\log x\le x-1$ for $x\in (0,\infty)$, we obtain that
\begin{align*}
\tr\big(\id-A_s(n)A_s(n)^*\big)=\sum_{j=1}^m 1-\mu_j
\le-\sum_{j=1}^m \log \mu_j
=-\log\det\big(A_s(n)A_s(n)^*\big).
\end{align*}
This allows us to conclude from \eqref{eq:finalmatrix2} and \eqref{eq:aNamatrix} that
\begin{align*}
\sum_{j=1}^M& m_j\left(\frac{1}{k_j^2}-k_j^2+\log|k_j|^4\right) +\sum_{n\in\Z}\tr\big(B(n)^2\big)
+2\sum_{n\in\Z}\tr\big(A(n)A(n)^*-\id\big)-\log\det\big(A(n)A(n)^*\big)\\
&\ge e_2+\dots+e_s+ f_2+\dots+f_s\,.
\end{align*}
All the errors $e_2,\dots,e_s$ and $f_2,\dots,f_s$ can be made arbitrarily small such that we get the desired result \eqref{eq:finalmatrix}.
Using an approximation argument, this result extends to Jacobi operators $W$ which only satisfy that the right-hand-side of \eqref{eq:finalmatrix} is finite. 


\section{The inequality is sharp}\label{sec:sharp}
We aim to show that inequality \eqref{eq:final} is sharp.
In \cite{Spiridonov1995} (see also \cite{Spiridonov1995b}) reflectionless potentials for the Jacobi operator were constructed using Darboux transformations. The obtained operators are of the form 
\begin{align*}
(\widetilde{W}u)(n)=\widetilde{a}(n+1)u(n+1)+u(n-1)+\widetilde{b}(n)u(n)
\end{align*}
which can easily be transformed to the Hermitian form \eqref{eq:defH}. Let $\omega>0$ and define $c_n:=\cosh(\omega n)$. 
The results of Spiridonov and Zhedanov let us conclude that the Jacobi operator $W_1$ of the form \eqref{eq:defH} with potential
\begin{align*}
b_1(n)=\frac{c_n}{c_{n+1}}-\frac{c_{n-1}}{c_{n}}
\end{align*}
and off-diagonal entries 
\begin{align*}
a_1(n)=-\frac{\sqrt{c_nc_{n+2}}}{c_{n+1}}
\end{align*}
is reflectionless and has one eigenvalue $\lambda_1=-2\cosh(\omega)<-2$. Since $k_1=\e^\omega$ we observe that
\begin{align}
k_1^2-\frac{1}{k_1^2}-\log|k_1|^4=\e^{2\omega}-\e^{-2\omega}-4\omega
\label{eq:sharp1k} 
\end{align}
and it remains to prove that the right-hand-side of \eqref{eq:final} converges to this value. Note that $W_1$ can also be obtained by applying the commutation method to the free discrete Schr{\"o}dinger operator $W_0$ with $a\equiv-1$ and $b\equiv0$. To this end we observe that $\varphi(n)=\cosh(\omega n)\notin\ell^2(\Z)$ is a solution to the equation $W_0\varphi=\lambda_1\varphi$. As shown by Gesztesy and Teschl \cite{Gesztesy1996}, the Jacobi operator with potential given by \eqref{eq:defV1} and off diagonal entries \eqref{eq:defa1} now has exactly one eigenvalue $\lambda_1$. It is easy to check that this operator is precisely the operator $W_1$ since
\begin{align*}
b_1(n)=\frac{c_n}{c_{n+1}}-\frac{c_{n-1}}{c_{n}}=\frac{c_n}{c_{n+1}}+\frac{c_{n+1}}{c_n}+\lambda_1\,.
\end{align*}
Note that the difference to our previous work is that here we have started with a real number $\lambda_1$ which is no eigenvalue of $W_0$ and have created an operator $W_1$ that has exactly one eigenvalue $\lambda_1$. In the previous sections we started with the ground state of a Jacobi operator and eliminated this point from the spectrum by means of exactly the same method. 

Using \eqref{eq:V1square} and repeating the calculations thereafter together with the facts that $\lim_{N\to\infty}\frac{c_N}{c_{N+1}}=\e^{-\omega}$ and $\lim_{N\to\infty}\frac{c_{-N-1}}{c_{-N}}=\e^\omega$, we can show that
\begin{align}
\sum_{n\in\Z}b_1(n)^2=-\e^{-2\omega}+\e^{2\omega}+2\sum_{n\in\Z}(1-a_1(n)^2)\,.
\label{eq:asquare1} 
\end{align}
It is easy to see that $\sum_{n\in\Z}\log a_1(n)^2$ is a telescoping series such that
\begin{align*}
\sum_{n=-N}^N\log a_1(n)^2
=\sum_{n=-N}^N\log\left(\frac{c_n}{c_{n+1}}\right)-\log\left(\frac{c_{n+1}}{c_{n+2}}\right)
=\log\left(\frac{c_{-N}}{c_{-N+1}}\right)-\log\left(\frac{c_{N+1}}{c_{N+2}}\right)
\end{align*}
which converges to $2\omega$ as $N\to\infty$. Thus
\begin{align*}
2\sum_{n\in\Z}\log a_1(n)^2=4\omega
\end{align*}
and together with \eqref{eq:sharp1k} and \eqref{eq:asquare1} we can conclude that for this special reflectionless operator $W_1$ there is equality in \eqref{eq:final}. 

This result points out an important similarity between \eqref{eq:final} and the Lieb--Thirring inequality on the continuum for the power $\gamma=\frac32$. Both inequalities are sharp for reflectionless potentials.  


\section{Inequalities for higher powers of eigenvalues and potential}\label{sec:higher}

Adapting the Aizenman--Lieb principle \cite{Aizenman1978}, we aim to prove spectral inequalities which depend on higher powers of the eigenvalues $\lambda_j$ and the potential $B$. In this section we shall restrict ourselves to the case of a discrete Schr{\"o}dinger operator such that $A(n)=-\id$ for all $n\in\Z$. Let $\lambda_j^+$ for $j=1,\dots,M^+$ denote the eigenvalues of $W$  that are larger than 2 and let $1\le m_j^+\le m$ denote their respective multiplicities. In complete analogy we define $\lambda_j^-$ for $j=1,\dots,M^-$ as the eigenvalues which are smaller than $-2$ and denote their multiplicities by $m_j^-$. In the following computations we will at times consider the positive/negative eigenvalues of a discrete Schr{\"o}dinger operator $W$ with potential $B'$ different to $B$. Whenever we do so, we will denote these eigenvalues by $\lambda_j^\pm(B')$ to make a clear distinction.  

In \cite{Hundertmark2002} Hundertmark and Simon showed that the inequalities
\begin{align}
\sum_{j=1}^{M^\pm}m_j^\pm\big((\lambda_j^\pm)^2-4\big)^{\frac12}\le\sum_{n\in\Z}\tr\big(B(n)_\pm\big)
\label{eq:order1+} 
\end{align}
hold. 
As noted in \cite{Simon2011}, a straightforward calculation shows that
\begin{align}
\frac{1}{2}\left(k^2-\frac{1}{k^2}-\log|k|^4\right)=\int_2^{|\lambda|} (E^2-4)^{\frac12}\dd E
\label{eq:linkG32} 
\end{align}
for $-k-\frac{1}{k}=\lambda\in\R\setminus[-2,2]$ with $|k|>1$.
This raises the question of whether we can derive \eqref{eq:finalmatrix1} from the result by Hundertmark and Simon \eqref{eq:order1+} and whether we can then obtain similar inequalities for higher powers of the eigenvalues by considering iterated integrals of $(E^2-4)^{\frac12}$. To this end, we compute that for $\gamma>\frac12$ 
\begin{align*}
\sum_{j=1}^{M^+}m_j^+\int_{2}^{\lambda_j^+}(E^2-4)^{\frac12}(\lambda_j^+-E)^{\gamma-\frac32}\dd E
&=\int_{0}^{\infty}\sum_{j=1}^{M^+}m_j^+\big((\lambda_j^+-t)^2-4\big)^{\frac12}t^{\gamma-\frac32}\chi_{\{\lambda_j^+-t\ge2\}}(t)\dd t\\
&=\int_{0}^{\infty}\sum_{j=1}^{M^+}m_j^+\big(\lambda_j^+(B-t)^2-4\big)^{\frac12}t^{\gamma-\frac32}\chi_{\{\lambda_j^+(B-t)\ge2\}}(t)\dd t
\end{align*}
where we have used the substitution $t=\lambda_j^+-E$. Note that by the Min-Max principle $\lambda_j^+(B-t)\le\lambda_j^+\big((B-t)_+\big)$ and thus we obtain from \eqref{eq:order1+} that
\begin{align*}
\sum_{j=1}^{M^+}m_j^+\int_{2}^{\lambda_j^+}(E^2-4)^{\frac12}(\lambda_j^+-E)^{\gamma-\frac32}\dd E\le\sum_{n\in\Z}\int_{0}^{\infty}t^{\gamma-\frac32}\tr\big((B(n)-t)_+\big)\dd t\,.
\end{align*}
Let $\mu_1(n),\dots,\mu_\ell(n)$  with $\ell(n)\le m$ be the eigenvalues of $B(n)$ which are positive. We observe that for fixed $n\in\Z$
\begin{align*}
\int_{0}^{\infty}t^{\gamma-\frac32}\tr\big((B(n)-t)_+\big)\dd t
=\sum_{k=1}^{\ell(n)}\int_0^\infty t^{\gamma-\frac32} \big(\mu_k(n)-t\big)_+\dd t
&=\sum_{k=1}^{\ell(n)}\mu_k(n)^{\gamma+\frac12}\int_0^1 s^{\gamma-\frac32}(1-s)\dd s\\&
=\B\left(\gamma-\frac12,2\right)\tr\left(B(n)_+^{\gamma+\frac12}\right)
\end{align*}
with the Beta function $\B(x,y)=\frac{\Gamma(x)\Gamma(y)}{\Gamma(x+y)}$ and the simplified notation $B(n)_+^{\gamma+1/2}=\big(B(n)_+\big)^{\gamma+1/2}$. 
As a consequence it holds that
\begin{align*}
\sum_{j=1}^{M^+}m_j^+\int_{2}^{\lambda_j^+}(E^2-4)^{\frac12}(\lambda_j^+-E)^{\gamma-\frac32}\dd E
\le \B\left(\gamma-\frac12,2\right)\sum_{n\in\Z}\tr\left(B(n)_+^{\gamma+\frac12}\right)
\end{align*}
and in complete analogy it is possible to show that
\begin{align*}
\sum_{j=1}^{M^-}m_j^-\int_{2}^{|\lambda_j^-|}(E^2-4)^{\frac12}(|\lambda_j^-|-E)^{\gamma-\frac32}\dd E
\le \B\left(\gamma-\frac12,2\right)\sum_{n\in\Z}\tr\left(B(n)_-^{\gamma+\frac12}\right).
\end{align*}
Together these two results prove the following theorem.

\begin{theorem}
Let $W$ be the discrete Schr{\"o}dinger operator \eqref{eq:defHmatrix} with $A(n)=-\id$ for all $n\in\Z$. Enumerate the eigenvalues outside of $[-2,2]$ as $|\lambda_1|\ge|\lambda_2|\ge\dots\ge2$ with multiplicities $m_j\le m$. Then it holds that for $\gamma>\frac12$
\begin{align}
\sum_{j} m_j\int_{2}^{|\lambda_j|}(E^2-4)^{\frac12}(|\lambda_j|-E)^{\gamma-\frac32}\dd E
\le \B\left(\gamma-\frac12,2\right)\sum_{n\in\Z}\tr \left(|B(n)|^{\gamma+\frac12}\right).
\label{eq:orderalpha} 
\end{align}
provided that the right-hand-side is finite. 
Defining $|k_j|\ge1$ via $\lambda_j=-k_j-\frac{1}{k_j}$, this inequality coincides for $\gamma=\frac32$ by \eqref{eq:linkG32} with \eqref{eq:finalmatrix1}. Exemplarily for  $\gamma=\frac52$ and $\gamma=\frac72$ the inequality can be written as
\begin{align*}
\sum_{j}m_j \left(\frac32\left(|k_j|-\frac{1}{|k_j|}\right)+\frac16\left(|k_j|^3-\frac{1}{|k_j|^3}\right)-\log|k_j|^2\left(|k_j|+\frac{1}{|k_j|}\right)\right)
&\le\frac16\sum_{n\in\Z}\tr \left(|B(n)|^{3}\right),\\
\sum_{j}m_j \left(\frac{1}{12}\left(k_j^4-\frac{1}{k_j^4}\right)+\frac73\left(k_j^2-\frac{1}{k_j^2}\right)
-\log|k_j|^2\left(k_j^2+\frac{1}{k_j^2}\right)-\log|k_j|^6\right)
&\le\frac{1}{12}\sum_{n\in\Z}\tr \left(|B(n)|^{4}\right).
\end{align*}

\end{theorem}

For half-integers values $\gamma\in\N+\frac12$ the substitution $E=t+1/t$ allows us to explicitly compute the left-hand-side of \eqref{eq:orderalpha} in terms of  $k_j$ and we observe the emergence of logarithmic terms of $k_j$. In these cases it is also possible to compute the left-hand-side of the inequality in terms of $\lambda_j$ and exemplarily for $\gamma=\frac32,\frac52$ we get 
\begin{align*}
\sum_{j}m_j\!
\left(\!2\log2+\frac12|\lambda_j|(\lambda_j^2-4)^\frac12 -2\log\big(|\lambda_j|+(\lambda_j^2-4)^\frac12\big)\!\right)
\!&\le\frac12\sum_{n\in\Z}\tr\left(|B(n)|^2\right),\\
\sum_{j}m_j\!
\left(\!2|\lambda_j|\log2+\frac12\lambda_j^2(\lambda_j^2-4)^\frac12\! -2|\lambda_j|\log\big(|\lambda_j|+(\lambda_j^2-4)^\frac12\big)\!-\frac13(\lambda_j^2-4)^\frac32\!\right)
\!&\le\frac16\sum_{n\in\Z}\tr\left(|B(n)|^3\right).
\end{align*}

In the scalar case $m=1$ we can use the same method as in \cite{Hundertmark2002} to pass from inequalities in the case $a\equiv-1$ to the general case where $a(n)\to-1$ as $|n|\to\infty$. This yields the following result.

\begin{theorem}
Let $W$ be the scalar Jacobi operator \eqref{eq:defH} with $\lim_{|n|\to\infty}a(n)=-1$. Enumerate the eigenvalues outside of $[-2,2]$ as $|\lambda_1|\ge|\lambda_2|\ge\dots\ge2$ . Then it holds that for  $\gamma>\frac12$
\begin{align*}
\sum_{j} \int_2^{|\lambda_j|}\!(E^2-4)^{\frac12}(|\lambda_j|-E)^{\gamma-\frac32}\dd E 
\le 3^{\gamma-\frac12} B\!\left(\gamma-\frac12,2\right)\!\!\left(\sum_{n\in\Z}|b(n)|^{\gamma+\frac12}+4\sum_{n\in\Z}|a(n)+1|^{\gamma+\frac12}\right)\!.
\end{align*}

\end{theorem}

We will now show some general properties of the left-hand-side of \eqref{eq:orderalpha}. We are interested in its behaviour for large eigenvalues as well as for eigenvalues which are close to the essential spectrum. To this end we define the function $G_\gamma(\lambda)$ for $\lambda>2$  and $\gamma>\frac12$ as
\begin{align*}
G_\gamma(\lambda)=\int_{2}^{\lambda}(E^2-4)^{\frac12}(\lambda-E)^{\gamma-\frac32}\dd E
\end{align*}
which allows us write the left-hand-side of \eqref{eq:orderalpha} simply as $\sum_{j}m_j G_\gamma(|\lambda_j|)$. 

\begin{proposition}
For $\lambda\to2$ we observe the asymptotic behaviour
\begin{align}
G_\gamma(\lambda)=2 \B\left(\gamma-\frac12,\frac32\right)(\lambda-2)^{\gamma}
+\Oo\big((\lambda-2)^{\gamma+1}\big).
\label{eq:asymp2} 
\end{align}
It furthermore holds that
\begin{align}
G_\gamma(\lambda)\ge 2\B\left(\gamma-\frac12,\frac32\right)(\lambda-2)^{\gamma}
\label{eq:compare} 
\end{align}
as well as
\begin{align}
G_\gamma(\lambda)\ge \B\left(\gamma-\frac12,2\right)(\lambda-2)^{\gamma+\frac12}\,.
\label{eq:compare2} 
\end{align}
For $\lambda\to\infty$ the asymptotic behaviour
\begin{align}
G_\gamma(\lambda)=\B\left(\gamma-\frac12,2\right)\lambda^{\gamma+\frac12}+
\Oo\big(\lambda^{\gamma-\frac32}\big)
\label{eq:asympinfty} 
\end{align}
is satisfied and
\begin{align}
G_\gamma(\lambda)\le \B\left(\gamma-\frac12,2\right) \lambda^{\gamma+\frac12}\,.
\label{eq:boundinfty} 
\end{align}

\end{proposition}

\begin{proof}
Substituting $t=(E-2)/(\lambda-2)$ in the definition of $G_\gamma$ shows that
\begin{align*}
G_\gamma(\lambda)
=(\lambda-2)^{\gamma-\frac12}\int_0^1\left(\big(t(\lambda-2)+2\big)^2-4\right)^\frac12(1-t)^{\gamma-\frac32}
\dd t
\end{align*} 
and the claimed asymptotic behaviour for $\lambda\to2$ now follows from
\begin{align*}
\left(\big(t(\lambda-2)+2\big)^2-4\right)^\frac12=2(\lambda-2)^\frac12t^\frac12+\Oo\big((\lambda-2)^{\frac32}\big).
\end{align*}
Using the fact that $(E^2-4)\ge 4(E-2)$ and the same substitution as above we can furthermore compute that
\begin{align*}
G_\gamma(\lambda)\ge2\int_2^\lambda (E-2)^\frac12(\lambda-E)^{\gamma-\frac32}\dd E
=(\lambda-2)^{\gamma}\int_0^1t^\frac12 (1-t)^{\gamma-\frac32}\dd t
\end{align*}
which is precisely \eqref{eq:compare}. Similarly we can apply the inequality $(E^2-4)\ge (E^2-2)^2$ and the above substitution to establish that
\begin{align*}
G_\gamma(\lambda)\ge\int_2^\lambda (E^2-2)(\lambda-E)^{\gamma-\frac32}\dd E
=(\lambda-2)^{\gamma+\frac12}\int_0^1 t (1-t)^{\gamma-\frac32}\dd t
\end{align*} 
which proves \eqref{eq:compare2}.
Substituting $t=E\lambda$ in the definition of $G_\gamma$ allows us to conclude that
\begin{align*}
G_\gamma(\lambda)
=\lambda^{\gamma+\frac12}\int_{\frac{2}{\lambda}}^1\left(t^2-\frac{4}{\lambda^2}\right)^{\frac12}(1-t)^{\gamma-\frac32}\dd t
\end{align*}
which can be used to show both \eqref{eq:boundinfty} and  \eqref{eq:asympinfty}.
\end{proof}

The asymptotic behaviour \eqref{eq:asympinfty} proves that the inequalities \eqref{eq:orderalpha} are optimal for large coupling. For a scalar potential $b$ we define a reordering $b^\pm(n)$ such that $b^+(1)\ge b^+(2)\ge\dots\ge0$ and $b^-(1)\le b^-(2)\le\dots\le0$. It then holds that
\begin{align*}
\lim_{\eta\to\infty}\frac1\eta \lambda_n^\pm(\eta b)=b^\pm(n)
\end{align*}
and consequently the ratio of the left-hand side and the right-hand-side of \eqref{eq:orderalpha} for potential $\eta b$ converges to 1 as $\eta\to\infty$. 

Inequalities \eqref{eq:compare} and \eqref{eq:compare2} allow us to compare our spectral inequality \eqref{eq:orderalpha} to the results of Hundertmark and Simon \cite{Hundertmark2002}. In their paper it is proven that in the case $A\equiv-\id$ the inequality
\begin{align}
\sum_{j} m_j\big(|\lambda_j|-2\big)^{\gamma}\le d_\gamma \sum_{n\in\Z}\tr \left(|B(n)|^{\gamma+\frac12}\right)
\label{eq:hsfree} 
\end{align}
with the constant
\begin{align*}
d_\gamma=\frac12 \frac{\Gamma(\gamma+1)\Gamma(2)}{\Gamma(\gamma+\frac32)\Gamma(\frac32)}
=\frac12 \frac{\B(\gamma-\frac12,2)}{\B(\gamma-\frac12,\frac32)}=2L^{\cl}_{\gamma,1}
\end{align*}
holds for all $\gamma\ge\frac12$.
Here, $L^{\cl}_{\gamma,1}$ is the semi-classical Lieb--Thirring constant. From \eqref{eq:compare} we can conclude that
\begin{align}
\frac{1}{d_\gamma}\sum_{j} m_j\big(|\lambda_j|-2\big)^{\gamma}
\le\frac{1}{\B(\gamma-\frac12,2)}\sum_{j} m_jG_{\gamma}(|\lambda_j|)
\le\sum_{n\in\Z}\tr\left(|B(n)|^{\gamma+\frac12}\right)
\label{eq:better1} 
\end{align}
which shows that at least in the case $A\equiv-\id$, inequality \eqref{eq:orderalpha} is indeed better than \eqref{eq:hsfree}. From \eqref{eq:asymp2} we can conclude that for eigenvalues $|\lambda_j|\to2$ approaching  the essential spectrum the two inequalities have the same asymptotic behaviour. 
In their paper \cite{Hundertmark2002}, Hundertmark and Simon also proved that for $\gamma\ge\frac12$
\begin{align}
\sum_{j}m_j (|\lambda_j|-2)^{\gamma+\frac12}\le\sum_{n\in\Z}\tr\left(|B(n)|^{\gamma+\frac12}\right)
\label{eq:hsfree2} 
\end{align}
which was found to be optimal for large coupling. Using \eqref{eq:compare2} we obtain that
\begin{align}
\sum_{j}m_j (|\lambda_j|-2)^{\gamma+\frac12}
\le\frac{1}{\B(\gamma-\frac12,2)}\sum_{j}m_jG_{\gamma}(|\lambda_j|)
\le\sum_{n\in\Z}\tr\left(|B(n)|^{\gamma+\frac12}\right)
\label{eq:better2} 
\end{align}
which leads to the conclusion that \eqref{eq:orderalpha} is also more precise than \eqref{eq:hsfree2}. From \eqref{eq:asympinfty} it follows that in the limit $|\lambda_j|\to\infty$ the two inequalities show the same asymptotic behaviour.

For choices of $\gamma$ which are not half-integers, the integral in the definition of $G_\gamma(\lambda)$ can unfortunately not be computed explicitly. We are  particularly interested in the case $\gamma=1$. 
We can at least numerically compute $G_{1}(\lambda)$ for different values of $\lambda$.  In Figure \ref{fig:plot12}  we compare the two ratios
\begin{align*}
R_1(\lambda)=\left(\frac{G_1(\lambda)}{\B(\frac12,2)}\right)\Big/\left(\frac{\lambda-2}{d_1}\right),
&&
R_2(\lambda)=\left(\frac{G_1(\lambda)}{\B(\frac12,2)}\right)\Big/(\lambda-2)^\frac32\,.
\end{align*}
From \eqref{eq:better1}, \eqref{eq:better2} it follows that $R_1(\lambda)\ge1$ and $R_2(\lambda)\ge1$ which is clearly visualised in the plot. The figure also indicates the asymptotic behaviours \eqref{eq:asymp2} and \eqref{eq:asympinfty}, which imply that $\lim_{\lambda\to2}R_1(\lambda)=\lim_{\lambda\to\infty}R_2(\lambda)=1$. 

\begin{figure}[!ht]
\begin{center}
\includegraphics[scale=0.42]{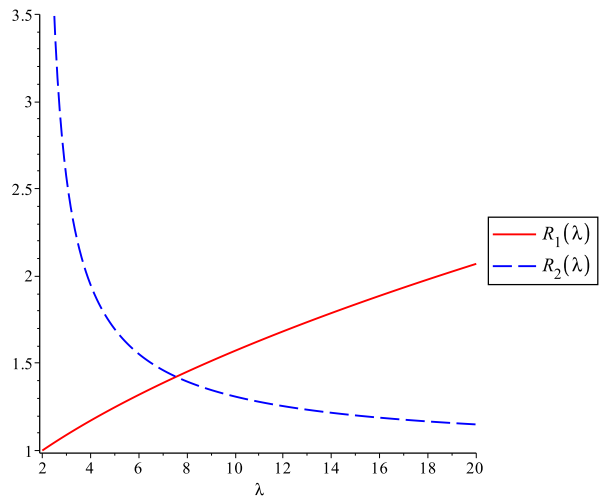}
\caption{Comparison of the governing terms in the spectral inequalities}
\label{fig:plot12}
\end{center}
\end{figure}


\section{Approximation of a Schr{\"o}dinger operator on the continuum by discrete Schr{\"o}dinger operators}
\label{sec:approx1}

In this section, we will use the established inequalities to prove analogous results for a Schr{\"o}dinger operator $-\frac{\dd^2}{\dd x^2}+V(x)$ on $L^2(\R)$. To approximate this operator by discrete Schr{\"o}dinger operators, we follow \cite{Wolff2001} where the following theorem is proven. 

\begin{theorem}{\cite[Theorem 2.3]{Wolff2001}}\label{th:approxdiscr}
Let $(E,\norm{\cdot})$ be a Banach Space and let $E_1$ be a dense linear subspace. Let $(F_k,\norm[k]{\cdot})$ be a sequence of Banach spaces such that for every $k\in\N$ there exists a (not necessarily bounded) linear operator $P_k:E_1\to F_k$ such that $\lim_{k\to\infty}\norm[k]{P_kf}=\norm{f}$ for all $f\in E_1$. Consider a densely defined closed operator $\big(H,\dom(H)\big)$ on $E$ with domain $\dom(H)$ and let $E_0\subset E_1$ be a core of $H$ such that $H(E_0)\subset E_1$. For each $k\in\N$ let $\big(H_k,\dom(H_k)\big)$ be a densely defined operator on $F_k$ such that $P_k(E_1)\subset \dom(H_k)$ and $\lim_{k\to\infty}\norm[k]{H_kP_kf-P_kHf}=0$ for all $f\in E_0$.  
If there exists an $M>0$ such that
\begin{align}
\norm{(H_k-\lambda)^{-1}}\le \frac{M}{\mathrm{dist}\big(\lambda,\sigma(H_k)\big)}
\label{eq:Mcond} 
\end{align} 
for all $\lambda\in\rho(H_k)$ and all $k\in\N$, then for any compact set $K\subset\C$
\begin{align*}
\lim_{k\to\infty} \mathrm{dist}\big(\sigma_a(H)\cap K,\sigma(H_k)\big)=0\,.
\end{align*}

\end{theorem}

Here, the distance between a bounded set $X\subset\C$ and a nonempty set $Y\subset\C$ is defined by
\begin{align*}
\mathrm{dist}\big(X,Y\big)=\sup_{x\in X}\inf_{y\in Y}|x-y|
\end{align*}
and $\sigma_a(H)$ denotes the approximate point spectrum of $H$. 

In our case $E=L^2(\R)$ and $H=-\frac{\dd^2}{\dd x^2}+V(x)$. For the moment we shall assume that the potential is smooth and compactly supported, $V\in\mathcal{C}_c^\infty(\R)$. The operator $H$  is self-adjoint on the domain $\dom(H)=H^2(\R)$, the Sobolev space of functions in $L^2(\R)$ that are twice weakly differentiable. We choose $E_1=E_0=\mathcal{C}^\infty_c(\R)$. 
which is well-known to form a core of $H$ and clearly also $H(E_0)\subset E_1$.  The Banach spaces $F_k$ are chosen to be $\ell^2(\Z)$ with the scalar product
\begin{align*}
(c,d)_k=\frac{1}{k}\sum_{n\in\Z} c(n) \overline{d(n)}\,.
\end{align*}
and the operator $P_k:E_1\to F_k$ is defined via
\begin{align*}
(P_kf)(n)=f\left(\frac nk\right),\, n\in\Z
\end{align*}
for all $f\in E_1$. Note that this is well-defined as the compact support of $f$ guarantees that $\norm[k]{P_kf}$ is given by a finite sum. Furthermore, by the definition of Riemann integrability, it holds that
\begin{align*}
\lim_{k\to\infty}\norm[k]{P_kf}^2=\lim_{k\to\infty}\frac1k\sum_{n\in\Z}\left|f\left(\frac nk\right)\right|^2
=\int_{\R}|f(x)|^2\dd x=\norm[L^2]{f}^2\,.
\end{align*} 
The operator $H_k$ is now defined on $F_k$ as
\begin{align*}
(H_ku)(n)=k^2\big(-u(n+1)+2u(n)-u(n-1)\big)+V\left(\frac nk\right)u(n)\,.
\end{align*}
Due to the compact support of the potential, $H_k$ is a bounded, self-adjoint operator on $F_k$. Thus the operator is in particular normal and hence \eqref{eq:Mcond} holds with $M=1$. It remains to show that $\lim_{k\to\infty}\norm[k]{H_kP_kf-P_kHf}=0$ for all $f\in E_0$. Let $N\in\N$ be sufficiently big such that $\mathrm{supp}(f)\subset [-N,N]$. We observe that
\begin{align*}
\norm[k]{H_kP_kf-P_kHf}^2=\frac1k\sum_{n=-Nk-1}^{Nk+1}
\left|k^2\left(
-f\left(\frac{n+1}{k}\right)+2f\left(\frac n k\right)-f\left(\frac{n-1}{k}\right)
\right)+f''\left(\frac nk\right)\right|^2\,.
\end{align*}
As $f\in\mathcal{C}_c^\infty(\R)$, Taylor's theorem shows that there exists a constant $C>0$ depending only on $f$ such that every summand in the above sum is smaller than $C \frac{1}{k^4}$. This proves that
\begin{align*}
\norm[k]{H_kP_kf-P_kHf}^2\le \frac{1}{k^5}(2Nk+3)
\end{align*}
where the right-hand-side converges to zero as $k\to\infty$. Thus all the assumptions of Theorem \ref{th:approxdiscr} are satisfied. 

Denote by $\mu_1<\dots<\mu_N<0$  the finite sequence of negative eigenvalues of $H$ in increasing order. Choosing $K=[\mu_1,\mu_N]\subset\R$,  Theorem  \ref{th:approxdiscr} states that 
\begin{align}
\lim_{k\to\infty}\max_{j=1,\dots,N}\inf_{\lambda\in \sigma(H_k)}|\mu_j-\lambda|=0\,.
\label{eq:approxdist} 
\end{align}
Note that $W_k=k^{-2}H_k-2$ is a discrete Schr{\"o}dinger operator of the form \eqref{eq:defH} with $a_k\equiv-1$ and scalar potential $b_k(n)=k^{-2}V(n/k)$. Its essential spectrum is given by $[-2,2]$ and it has a finite sequence of eigenvalues outside of this interval. These eigenvalues satisfy the spectral inequalities discussed in the previous section, in particular \eqref{eq:order1+} and \eqref{eq:orderalpha}. 
If we denote by $\lambda_{k,j}$ the discrete eigenvalues of $H_k$ the above observations let us conclude that each $\lambda_{k,j}$ is in the complement of the interval $k^2[0,4]$ and that
\begin{align*}
\sum_{j}\big((k^{-2}\lambda_{k,j}-2)^2-4\big)^{\frac12}\le
\frac{1}{k^2}\sum_{n\in\Z}\left|V\left(\frac nk\right)\right|
\end{align*}
as well as 
\begin{align*}
\sum_{j}G_{\gamma}(|k^{-2}\lambda_{k,j}-2|)
\le \B\left(\gamma-\frac12,2\right)\frac{1}{k^{2\gamma+1}}\sum_{n\in\Z}\left|V\left(\frac nk\right)\right|^{\gamma+\frac12}
\end{align*}
for all $\gamma>\frac12$.
If we sum only over all the negative eigenvalues $\lambda_{k,j}^-$ of $H_k$ then 
\begin{align}
\sum_{j}\big((k^{-2}\lambda_{k,j}^--2)^2-4\big)^{\frac12}\le
\frac{1}{k^2}\sum_{n\in\Z}V\left(\frac nk\right)_-
\label{eq:approx12} 
\end{align}
and for $\gamma>\frac12$
\begin{align}
\sum_{j}G_{\gamma}(|k^{-2}\lambda_{k,j}^--2|)
\le \B\left(\gamma-\frac12,2\right)\frac{1}{k^{2\gamma+1}}\sum_{n\in\Z}V\left(\frac nk\right)_-^{\gamma+\frac12}\,.
\label{eq:approxp} 
\end{align}
A simple calculation for negative $\lambda<0$ shows that
\begin{align*}
\big((k^{-2}\lambda-2)^2-4\big)^\frac12=k^{-1}(k^{-2}\lambda^2-4\lambda)^{\frac12} 
\end{align*}
and consequently inequality \eqref{eq:approx12} can be written as
\begin{align*}
\sum_{j}\big(k^{-2}(\lambda_{k,j}^-)^2-4\lambda_{k,j}\big)^{\frac12} \le
\frac{1}{k}\sum_{n\in\Z}V\left(\frac nk\right)_-\,.
\end{align*}
From \eqref{eq:approxdist} we conclude that the negative eigenvalues $\lambda_{k,j}^-$ approximate the eigenvalues $\mu_1,\dots,\mu_N$ as $k\to\infty$. Thus, taking the limit $k\to\infty$ in the above spectral inequality and using the definition of Riemann integrability we arrive at
\begin{align*}
\sum_{j}|\mu_j|^\frac12\le
\frac{1}{2}\int_{\R}V(x)_-\dd x
\end{align*}
which is the best possible bound. Using a further approximation argument, the assumption $V\in\mathcal{C}_c^\infty(\R)$ can be generalised to the weaker $V\in L^1(\R)$.

Using the substitution $s=k^2(E-2)/|\lambda|$ it is straightforward to show that for $\gamma>\frac12$ and $\lambda<0$ we can write 
\begin{align*}
G_{\gamma}(|k^{-2}\lambda-2|)
=G_\gamma(2+k^{-2}|\lambda|)
=\frac{1}{k^{2\gamma}}|\lambda|^\gamma\int_0^1s^\frac12(1-s)^{\gamma-\frac32}(4+sk^{-2}|\lambda|)^\frac12\dd s\,.
\end{align*}
Together with \eqref{eq:approxp}  this allows us to conclude that 
\begin{align*}
\sum_{j}|\lambda_{k,j}^-|^\gamma\int_0^1s^\frac12(1-s)^{\gamma-\frac32}(4+sk^{-2}|\lambda_{k,j}^-|)^\frac12\dd s
\le \B\left(\gamma-\frac12,2\right)\frac{1}{k}\sum_{n\in\Z}V\left(\frac nk\right)_-^{\gamma+\frac12}\,.
\end{align*}
In the limit $k\to\infty$ this inequality yields 
\begin{align*}
\sum_{j}|\mu_j|^\gamma\le\frac{\B(\gamma-\frac12,2)}{2\B(\gamma-\frac12,\frac32)}
\int_{\R}V(x)_-^{\gamma+\frac12}\dd x
=2L^{\cl}_{\gamma,1}\int_{\R}V(x)_-^{\gamma+\frac12}\dd x
\end{align*}
which is not the sharp result. Note that using \eqref{eq:hs1} instead of \eqref{eq:orderalpha} leads to the same result, as the difference between thews two inequalities vanishes for small eigenvalues. As mentioned before, for $\gamma\ge1$ the constant $d_\gamma$ can be improved to $\frac{\pi}{\sqrt{3}}L^\cl_{\gamma,1}$ \cite{Sahovic2010}. This yields the result of \cite{Eden1991} on $L^2(\R)$. In order to get a sharp result, we will use a different approximation of $H$ in the next section. 


\section{Approximation of a Schr{\"o}dinger operator on the continuum by Jacobi operators}
\label{sec:approx2}
One could expect better results by approximating the operator on the continuum by a discrete operator whose off-diagonal entries $a_k(n)$ are not restricted to be $-1$.  
Let the vector spaces $E_0,E_1,F_k$ and the linear operators $P_k$ be defined as in the previous section. Take $c\in[0,1]$ and set $d=(1-c)/2$ such that $c+2d=1$. We now define the operator $H_k$ on $F_k$ as
\begin{align*}
(H_ku)(n)
=&k^2\left(-u(n-1)+2u(n)-u(n+1)\right)\\
&+c V\left(\frac nk\right)u(n)+dV\left(\frac{n-1}k\right)u(n-1)+dV\left(\frac nk\right)u(n+1)\,.
\end{align*}
Suppose that $f\in E_0$ and that $N\in\N$ is sufficiently large such that $[-N,N]$ covers both the support of $V$ and the support of $f$. Then it holds that
\begin{align*}
&\norm[k]{H_kP_kf-P_kHf}^2\le\frac1k\sum_{n=-Nk+1}^{Nk+1}
\Bigg(\left|k^2\left(
-f\left(\frac{n+1}{k}\right)+2f\left(\frac n k\right)-f\left(\frac{n-1}{k}\right)
\right)+f''\left(\frac nk\right)
\right|\\
&+\left|c V\left(\frac nk\right)f\left(\frac{n}{k}\right)+dV\left(\frac{n-1}k\right)f\left(\frac{n-1}{k}\right)+dV\left(\frac nk\right)f\left(\frac{n+1}{k}\right)
-V\left(\frac nk\right)f\left(\frac nk\right)\right|\Bigg)^2.
\end{align*}
As discussed previously, Taylor's theorem and the boundedness of $f$ yield the existence of a constant $C_1$ such that
\begin{align}
\left|k^2\left(
-f\left(\frac{n+1}{k}\right)+2f\left(\frac n k\right)-f\left(\frac{n-1}{k}\right)
\right)+f''\left(\frac nk\right)
\right|\le C_1\frac{1}{k^2}\,.
\label{eq:approxest1} 
\end{align}
To find a similar bound for the second sum, we first note that by Taylor's theorem there exists $\xi\in (n/k-1/k,n/k)$ such that
\begin{align*}
V\left(\frac{n-1}k\right)=V\left(\frac nk\right)+\frac1kV'(\xi)\,.
\end{align*}
This lets us conclude that
\begin{align}
\begin{split}
&\left|c V\left(\frac nk\right)f\left(\frac{n}{k}\right)+dV\left(\frac{n-1}k\right)f\left(\frac{n-1}{k}\right)+dV\left(\frac nk\right)f\left(\frac{n+1}{k}\right)
-V\left(\frac nk\right)f\left(\frac nk\right)\right|\\
&\le\left| V\left(\frac nk\right)\right|\left|cf\left(\frac{n}{k}\right)+df\left(\frac{n-1}{k}\right)+df\left(\frac{n+1}{k}\right)
-f\left(\frac nk\right)\right|
+\frac{1}{k}\left|V'(\xi)f\left(\frac{n-1}{k}\right)\right|,
\end{split}
\label{eq:approxest2} 
\end{align}
where the last term can be bounded by $C_2/k$ with the constant $C_2>0$ depending only on $V$ and $f$ since $V,f \in\mathcal{C}_c^\infty(\R)$. 
By Taylor's theorem there also exist $\xi_-\in(n/k-1/k,n/k)$ and $\xi_+\in(n/k,n/k+1/k)$ such that 
\begin{align*}
f\left(\frac{n\pm1}k\right)=f\left(\frac nk\right)\pm\frac1kf'(\xi_\pm)\,.
\end{align*}
This allows us to continue \eqref{eq:approxest2} as
\begin{align*}
&\left|c V\left(\frac nk\right)f\left(\frac{n}{k}\right)+dV\left(\frac{n-1}k\right)f\left(\frac{n-1}{k}\right)+dV\left(\frac nk\right)f\left(\frac{n+1}{k}\right)
-V\left(\frac nk\right)f\left(\frac nk\right)\right|\\
&\le\left| V\left(\frac nk\right)\right|\frac{d}{k}\left|f'(\xi_+)-f'(\xi)_-\right|+C_2\frac1k\\
&\le C_3\frac1k
\end{align*}
where $C_3>0$ is a constant depending only on $V$ and $f$. Together with \eqref{eq:approxest1} we have shown that
\begin{align*}
\norm[k]{H_kP_kf-P_kHf}^2\le\frac{1}{k}(2Nk+3)\left(C_1\frac{1}{k^2}+C_3\frac1k\right)^2
\end{align*} 
where the right-hand-side clearly converges to zero as $k\to\infty$. Thus all the assumptions of Theorem \ref{th:approxdiscr} are satisfied. 

The operator $W_k=k^{-2}H_k-2$ is of the form \eqref{eq:defH} with potential
$b_k(n)=ck^{-2}V(n/k)$ and off-diagonal terms $a_k(n)=-1+dk^{-2}V(n/k)$. Note that $a_k(n)<0$ if only $k$ is sufficiently large.
Let $\lambda_{k,j}^-$ be the negative eigenvalues of $H_k$. Inequality \eqref{eq:final} now holds for the eigenvalues of $W_k$ and in particular if we only consider the negative eigenvalues. Recalling \eqref{eq:linkG32} this can be written as 
\begin{align*}
2\sum_{j}G_{\frac32}(|k^{-2}\lambda_{k,j}^--2|)\le\sum_{n\in\Z}b_k(n)^2+2\sum_{n\in\Z} \left(a_k(n)^2-1-\log a_k(n)^2\right)\,.
\end{align*}
We now apply \eqref{eq:compare} and insert the definitions of $a_k(n)$ and $b_k(n)$ to conclude that \begin{align*}
&4\B\left(1,\frac32\right)
\sum_j|\lambda_{k,j}^-|^{\frac32}\\\le &\frac{c^2}{k} \sum_{n\in\Z}V\left(\frac nk\right)^2
+2\sum_{n\in\Z}\frac{d^2}{k} V\left(\frac nk\right)^2-2dkV\left(\frac nk\right)
-k^3\log\left(1-\frac{2d}{k^2}V\left(\frac nk\right)+\frac{d^2}{k^4}V\left(\frac nk\right)^2\right).
\end{align*}
Using the Taylor expansion $\log(1+x)=x-x^2/2+\Oo(x^3)$ for small $|x|$ and \eqref{eq:approxdist} as well as the Riemann integrability of $V$, we can let $k\to\infty$ to obtain 
\begin{align*}
\sum_j|\mu_{j}|^{\frac32}\le \frac{1}{4B(1,\frac32)}(c^2+4d^2)\int_{\R}V(x)^{2}\dd x= 2L_{\frac32,1}^{\cl}(c^2+4d^2)\int_{\R}V(x)^{2}\dd x\,.
\end{align*}
We can minimise the right-hand-side with respect to $c\in[0,1]$ and it is straightforward to prove that the minimum is attained at $c=\frac12$. This yields  
\begin{align*}
\sum_j|\mu_{j}|^{\frac32}\le L_{\frac32,1}^{\cl}\int_{\R}V(x)^{2}\dd x
\end{align*}
which  by continuity can be shown to hold under the more general assumption $V\in L^2(\R)$. 
Finally, the Min-Max principle implies that $V$ on the right-hand-side of the inequality can be replaced by $V_-$, yielding the sharp Lieb--Thirring inequality 
\begin{align*}
\sum_j|\mu_{j}|^{\frac32}\le L_{\frac32,1}^{\cl}\int_{\R}V(x)_-^{2}\dd x\,.
\end{align*}
The Aizenman-Lieb principle extends this result immediately to higher powers $\gamma\ge\frac32$.

We can also try to obtain better bounds for $\frac12<\gamma<\frac32$ using this approximation and \eqref{eq:hs1}.
If we sum only over all the negative eigenvalues in \eqref{eq:hs1} then the constant 4 can be replaced by 2 and this yields that for all $\gamma\ge\frac12$
\begin{align*}
\sum_j|\lambda_{k,j}^-|^{\gamma}
&\le c_\gamma\left(\frac{d^{\gamma+\frac12}}{k}\sum_{n\in\Z}\left|V\left(\frac nk\right)\right|^{\gamma+\frac12}+\frac{2d^{\gamma+\frac12}}{k}\sum_{n\in\Z}\left|V\left(\frac nk\right)\right|^{\gamma+\frac12}\right)\\
&=c_\gamma\left(c^{\gamma+\frac12}+2d^{\gamma+\frac12}\right)
\frac{1}{k}\sum_{n\in\Z}\left|V\left(\frac nk\right)\right|^{\gamma+\frac12}.
\end{align*}
Letting $k\to\infty$ and using \eqref{eq:approxdist}, we arrive at
\begin{align*}
\sum_j|\mu_{j}|^{\gamma}\le 2L_{\gamma,1}^{\cl}3^{\gamma-\frac12}\left(c^{\gamma+\frac12}+2d^{\gamma+\frac12}\right)\int_{\R}|V(x)|^{\gamma+\frac12}\dd x\,.
\end{align*}
We can now minimise the function $c^{\gamma+\frac12}+2d^{\gamma+\frac12}$ on the right-hand-side with respect to $c\in[0,1]$. If $\gamma=\frac12$ we observe that the function is constant $1$ and consequently the minimum is attained at any $c\in[0,1]$. Thus, no matter which approximation we choose exactly, we will always get the sharp Lieb--Thirring constant $2L_{1/2,1}^{\cl}$. For $\gamma>\frac12$ we find that the minimum is attained at $c=\frac13$. Consequently we obtain that 
\begin{align*}
\sum_j|\mu_{j}|^{\gamma}\le 2L_{\gamma,1}^{\cl}\int_{\R}|V(x)|^{\gamma+\frac12}\dd x
\end{align*}
which is unfortunately exactly the same result as in the previous section. This is not surprising, as Hundertmark and Simon derived the inequalities \eqref{eq:hs1} from the result for constant diagonal terms using a perturbation argument. To obtain the best possible Lieb--Thirring constants for $\frac12<\gamma<\frac32$ from our approximation method, it would first be necessary to establish the correct terms depending on $a(n)$ on the right-hand-side of \eqref{eq:orderalpha} such that the inequality becomes sharp or to further improve inequality \eqref{eq:hs1} in the special case of a discrete Schr{\"o}dinger operator.

\section*{Acknowledgements}
The author is very thankful to his supervisor Ari Laptev for providing the initial idea behind this work and for his engagement in stimulating discussions. The author would also like to thank Alexander Veselov for help with the references.
Furthermore, funding through the Roth Studentship from the Department of Mathematics at Imperial College London is gratefully acknowledged.
\bibliographystyle{amsplain}
\bibliography{biblio_discretebl}
\end{document}